\RequirePackage{amsmath}
\documentclass[runningheads]{llncs}
\usepackage{multicol}
\usepackage{amsfonts}
\usepackage{amssymb}
\usepackage{graphicx}
\usepackage{epic}
\usepackage{eepic}
\usepackage{epsfig,float}
\usepackage{verbatim}
\usepackage{pdfsync}
\usepackage{gastex}
\usepackage[lowtilde]{url}
\usepackage{enumitem}
\usepackage{url} 

\usepackage{bbm}
\usepackage{xcolor}

\newcommand{\Vsf}{\mathbf{T}^{\le 5}} 
\newcommand{\Wsf}{\mathbf{T}^{\ge 6}}

\usepackage{xcolor}

\usepackage[utf8]{inputenc}\usepackage[T1]{fontenc}

\renewcommand{\le}{\leqslant}
\renewcommand{\ge}{\geqslant}

\newcommand{\eps}{\varepsilon}
\newcommand{\emp}{\emptyset}

\newcommand{\Sig}{\Sigma}

\newcommand{\noin}{\noindent}

\newcommand{\bi}{\begin{itemize}}
\newcommand{\ei}{\end{itemize}}
\newcommand{\be}{\begin{enumerate}}
\newcommand{\ee}{\end{enumerate}}
\newcommand{\bd}{\begin{description}}
\newcommand{\ed}{\end{description}}
\newcommand{\bq}{\begin{quote}}
\newcommand{\eq}{\end{quote}}

\newcommand{\cD}{{\mathcal D}}
\newcommand{\cE}{{\mathcal E}}

\newcommand{\cN}{{\mathcal N}}

\newcommand{\cT}{{\mathcal T}}

\newcommand{\lraL}{{\mathbin{\approx_L}}}

\title{Complexity of Left-Ideal, Suffix-Closed and Suffix-Free Regular Languages\thanks{This work was supported by the Natural Sciences and Engineering Research Council of Canada 
grant No.~OGP0000871.}
}

\author{Janusz~Brzozowski and Corwin Sinnamon}
\authorrunning{J. Brzozowski,  C. Sinnamon}  
\titlerunning{Complexity of Left-Ideal, Suffix-Closed and Suffix-Free Languages}

\institute{David R. Cheriton School of Computer Science, University of Waterloo, \\
Waterloo, ON, Canada N2L 3G1\\
{\tt \{brzozo@uwaterloo.ca, sinncore@gmail.com}
}

\begin{document}
\maketitle

\begin{abstract}
A language $L$ over an alphabet $\Sig$ is suffix-convex if, for any words $x,y,z\in\Sig^*$, whenever $z$ and $xyz$ are in $L$, then so is $yz$.
Suffix-convex languages include three special cases: left-ideal, suffix-closed, and suffix-free languages.
We examine complexity properties of these three special classes of suffix-convex regular languages.
In particular, we study the quotient/state complexity of boolean operations, product (concatenation), star, and reversal on these languages,  as well as the size of their syntactic semigroups, and the quotient complexity of their atoms.
\medskip

\noin
\end{abstract}

\section{Introduction} 
\noin
{\bf Suffix-Convex Languages}
Convex languages were introduced in 1973~\cite{Thi73}, and revisited in 2009~\cite{AnBr09}.
For $w,x,y \in \Sig^*$, if $w=xy$, then $y$ is a \emph{suffix} of $w$.
A~language $L$ is \emph{suffix-convex} if, whenever $z$ and $xyz$ are in $L$,  then $yz$ is also in $L$, for all $x,y,z \in \Sig^*$.
Suffix-convex languages include three well-known subclasses: left-ideal, suffix-closed, and suffix-free languages.
A language $L$ is a \emph{left ideal} if it is non-empty and satisfies the equation $L=\Sig^*L$.
Left ideals play a role in pattern matching: If one is searching for all words ending with words in some language $L$ in a given text (a word over $\Sig^*$), then one is looking for words in $\Sig^*L$. 
Left ideals also constitute a basic concept in semigroup theory.
A~language $L$ is \emph{suffix-closed} if, whenever $w$ is in $L$ and $x$ is a suffix of $w$,  then $x$ is also in $L$, for all $w,x\in \Sig^*$.
The complement of every suffix-closed language not equal to $\Sig^*$ is a left ideal.
A language is \emph{suffix-free} if no word in the language is a suffix of another word in the language.
Suffix-free languages (with the exception of $\{\eps\}$, where $\eps$ is the empty word) 
are suffix codes.
They 
have many applications, and have been studied extensively; see~\cite{BPR10} for example. 
\smallskip

\noin
{\bf Quotient/State Complexity}
If $\Sig$ is an alphabet and $L\subseteq\Sig^*$ is a language such that every letter of $\Sig$ appears in some word of $L$, then the \emph{(left) quotient} of $L$ by a word $w\in\Sig^*$ is $w^{-1}L=\{x\mid wx\in L\}$. A language is regular if and only if it has a finite number of distinct quotients. 
So the number of quotients of $L$, the 
\emph{quotient complexity} $\kappa(L)$~\cite{Brz10} of $L$, is a natural measure of complexity for $L$.
A concept equivalent to quotient complexity is the \emph{state complexity}~\cite{Yu01} of $L$, which is the number of states in a complete minimal deterministic finite automaton (DFA)  with alphabet $\Sig$ recognizing $L$.
We refer to quotient/state complexity simply as \emph{complexity}.

If $L_n$ is a regular language of complexity $n$, and $\circ$ is a unary operation, then 
the \emph{complexity  of $\circ$} 
is  the maximal value of $\kappa(L_n^\circ)$, expressed as a function of $n$,
as $L_n$ ranges over all regular languages of complexity $n$.
Similarly, if $L'_m$ and $L_n$ are regular languages of  complexities $m$ and $n$ respectively, $\circ$ is a binary operation, then
the \emph{complexity of $\circ$} 
is  the maximal value of $\kappa(L'_m \circ L_n)$, expressed as a function 
 of $m$ and $n$,
as $L'_m$ and $L_n$ range over all regular languages of complexities $m$ and $n$, respectively.
The complexity of an operation is a lower bound on its time and space complexities, and has been studied extensively; see~\cite{Brz10,Brz13,HoKu11,Yu01}.

In the past the complexity of a binary operation was studied under the assumption that the arguments of the operation are \emph{restricted} to be over the same alphabet, but this restriction was removed in~\cite{Brz16}. 
We study both the restricted and unrestricted cases.

\noin
{\bf Witnesses} To find the complexity of a unary operation we find an upper bound on this complexity, and  languages that meet this bound. 
We require a language $L_n$ for each $n\ge k$, that is, a sequence $(L_k, L_{k+1}, \dots)$, where $k$ is a small integer, because the bound may not hold for small values of $n$. 
Such a sequence is a \emph{stream} of languages. 
For a binary operation we require two streams. 
Sometimes the same stream can be used for both operands;
in general, however, this is not the case. 
For example, the bound for union is $mn$, and it cannot be met by languages from one stream if $m=n$ because $L_n \cup L_n=L_n$ and the complexity is $n$ instead of $n^2$.

\noin
{\bf Dialects}
For all common binary operations on regular languages the second stream can be a ``dialect'' of the first, that is it can ``differ only slightly'' from the first, and all the bounds can still be met~\cite{Brz13}. 
Let $\Sigma=\{a_1,\dots,a_k\}$ be an alphabet ordered as shown;
if $L\subseteq \Sigma^*$, we denote it by $L(a_1,\dots,a_k)$ to stress its dependence on $\Sigma$.
A \emph{dialect}  of $L$ is obtained by  deleting letters of $\Sigma$ in the words of $L$, or replacing them by letters of another alphabet $\Sigma'$.
More precisely, for  a partial  injective map $\pi \colon \Sigma \mapsto \Sigma'$,
we obtain a dialect of $L$ by replacing each letter $a \in \Sigma$ by $\pi(a)$ in every word of $L$,
or deleting the word entirely if $\pi(a)$ is undefined.
We write $L(\pi(a_1),\dots, \pi(a_k))$ to denote the dialect of $L(a_1,\dots,a_k)$ given by $\pi$,
and we denote undefined values of $\pi$ by  ``$-$''.
For example, if $L(a,b,c)= \{a, ab, ac\}$, then  $L(b,-,d)$ is the language $\{b,bd\}$.
Undefined values at the end of the alphabet are omitted. 
A similar definition applies to DFAs. 
Our definition of dialect is more general than that of~\cite{BDL15,BrSi16}, where only the case $\Sig'=\Sig$ was allowed.

\noin
{\bf Most Complex Streams} 
It was proved 
that there exists a stream $(L_3, L_4, \dots)$ of regular languages 
which together with some dialects
meets all the complexity bounds for reversal, (Kleene) star, product (concatenation), and all binary boolean operations~\cite{Brz13,Brz16}. 
Moreover, this stream meets two additional complexity bounds: the size of the syntactic semigroup, and the complexities of atoms (discussed later). 
A stream of deterministic finite automata (DFAs) corresponding to a most complex language stream is  a \emph{most complex DFA stream}.
In defining a most complex stream we try to minimize the size of the union of the alphabets of the dialects required to meet all the bounds.

Most complex streams are useful in the designs of systems dealing with regular languages and finite automata. 
To know the maximal sizes of automata that can be handled by the system it suffices to use the most complex stream to test all the operations.

It is known that there is a most complex stream of left ideals that meets all the bounds in both the restricted~\cite{BDL15,BrSi16} and unrestricted~\cite{BrSi16} cases, but a most complex suffix-free stream does not exist~\cite{BrSz15a}.
\medskip

\noin
{\bf Our Contributions}
\be
\item
We derive a new left-ideal stream from the most complex left-ideal stream and show that it meets all the complexity bounds except that for product. 
\item
We prove that the complement of the new left-ideal stream is a most complex  suffix-closed stream.
\item
We find a new suffix-free stream that meets the bounds for star, product and boolean operations; it has simpler transformations than the  known stream.
\item
Our witnesses for left-ideal, suffix-closed, and suffix-free streams are all derived from one most complex regular stream.
\ee

\section{Background}
\noin
{\bf Finite Automata}
A \emph{deterministic finite automaton (DFA)} is a quintuple
$\cD=(Q, \Sigma, \delta, q_0,F)$, where
$Q$ is a finite non-empty set of \emph{states},
$\Sig$ is a finite non-empty \emph{alphabet},
$\delta\colon Q\times \Sig\to Q$ is the \emph{transition function},
$q_0\in Q$ is the \emph{initial} state, and
$F\subseteq Q$ is the set of \emph{final} states.
We extend $\delta$ to a function $\delta\colon Q\times \Sig^*\to Q$ as usual.
A~DFA $\cD$ \emph{accepts} a word $w \in \Sigma^*$ if ${\delta}(q_0,w)\in F$. The language accepted by $\cD$ is denoted by $L(\cD)$. If $q$ is a state of $\cD$, then the language $L^q$ of $q$ is the language accepted by the DFA $(Q,\Sigma,\delta,q,F)$. 
A state is \emph{empty} if its language is empty. Two states $p$ and $q$ of $\cD$ are \emph{equivalent} if $L^p = L^q$. 
A state $q$ is \emph{reachable} if there exists $w\in\Sig^*$ such that $\delta(q_0,w)=q$.
A DFA is \emph{minimal} if all of its states are reachable and no two states are equivalent.
Usually DFAs are used to establish upper bounds on the complexity of operations and also as witnesses that meet these bounds.

A \emph{nondeterministic finite automaton (NFA)} is a quintuple
$\cD=(Q, \Sigma, \delta, I,F)$, where
$Q$,
$\Sig$ and $F$ are defined as in a DFA, 
$\delta\colon Q\times \Sig\to 2^Q$ is the \emph{transition function}, and
$I\subseteq Q$ is the \emph{set of initial states}. 
An \emph{$\eps$-NFA} is an NFA in which transitions under the empty word $\eps$ are also permitted.
\smallskip

\noin
{\bf Transformations}
We use $Q_n=\{0,\dots,n-1\}$ as our basic set with $n$ elements.
A \emph{transformation} of $Q_n$ is a mapping $t\colon Q_n\to Q_n$.
The \emph{image} of $q\in Q_n$ under $t$ is denoted by $qt$.
If $s$ and $t$ are transformations of $Q_n$, their composition is  denoted $(qs)t$ when applied to $q \in Q_n$.
Let $\cT_{Q_n}$ be the set of all $n^n$ transformations of $Q_n$; then $\cT_{Q_n}$ is a monoid under composition. 

For $k\ge 2$, a transformation (permutation) $t$ of a set $P=\{q_0,q_1,\ldots,q_{k-1}\} \subseteq Q$ is a \emph{$k$-cycle}
if $q_0t=q_1, q_1t=q_2,\ldots,q_{k-2}t=q_{k-1},q_{k-1}t=q_0$.
This $k$-cycle is denoted by $(q_0,q_1,\ldots,q_{k-1})$.
A~2-cycle $(q_0,q_1)$ is called a \emph{transposition}.
A transformation  that sends all the states of $P$ to $q$ and acts as the identity on the remaining states is denoted by $(P \to q)$ the transformation 
$(Q_n\to p)$ is called  \emph{constant}.  If $P=\{p\}$ we write  $( p\to q)$ for $(\{p\} \to q)$.
 The identity transformation is denoted by $\mathbbm 1$.
 The notation $(_i^j \; q\to q+1)$ denotes a transformation that sends $q$ to $q+1$ for $i\le q\le j$ and is the identity for the remaining states. the notation $(_i^j \; q\to q-1)$ is defined similarly.
\smallskip

\noin
{\bf Semigroups}
The \emph{Myhill congruence} $\lraL$~\cite{Myh57} (also known as the \emph{syntactic congruence})  of a language $L\subseteq \Sig^*$ is defined on $\Sig^+$ as follows:
For $x, y \in \Sig^+,  x \lraL y $  if and only if  $wxz\in L  \Leftrightarrow wyz\in L$ for all  $w,z \in\Sig^*.
$
The quotient set $\Sig^+/ \lraL$ of equivalence classes of  $\lraL$ is a semigroup, the \emph{syntactic semigroup} $T_L$ of $L$.

Let $\cD = (Q_{n}, \Sig, \delta,  0, F)$ be a DFA. For each word $w \in \Sig^*$, the transition function induces a transformation $\delta_w$ of $Q_n$ by  $w$: for all $q \in Q_n$, 
$q\delta_w = \delta(q, w).$ 
The set $T_{\cD}$ of all such transformations by non-empty words is the \emph{transition semigroup} of $\cD$ under composition~\cite{Pin97}. Sometimes we use the word $w$ to denote the transformation it induces; thus we write $qw$ instead of $q\delta_w$.
We extend the notation to sets: if $P\subseteq Q_n$, then $Pw=\{pw\mid p\in P\}$.
We also find  write $P\stackrel{w}{\longrightarrow} Pw$ to indicate that the image of $P$ under $w$ is $Pw$.

If  $\cD$ is a minimal DFA of $L$, then $T_{\cD}$ is isomorphic to the syntactic semigroup $T_L$ of $L$~\cite{Pin97}, and we represent elements of $T_L$ by transformations in~$T_{\cD}$. 
The size of this semigroup has been used as a measure of complexity~\cite{Brz13,BrYe11,HoKo04,KLS05}.
\smallskip

\noin
{\bf Atoms}
Atoms are defined by a left congruence, where two words $x$ and $y$ are equivalent if 
 $ux\in L$ if and only if  $uy\in L$ for all $u\in \Sig^*$. 
 Thus $x$ and $y$ are equivalent if
 $x\in u^{-1}L$ if and only if $y\in u^{-1}L$.
 An equivalence class of this relation is an \emph{atom} of $L$~\cite{BrTa14}. 
Thus an atom is a non-empty intersection of complemented and uncomplemented quotients of $L$. 
The number of atoms and their  complexities were suggested as possible measures of complexity of regular languages~\cite{Brz13},
because all the quotients of a language, and also the quotients of atoms, are always unions of atoms
~\cite{BrTa13,BrTa14,Iva16}.
\smallskip

\noin
{\bf Our Key Witness}
\label{sec:regular}
The stream $(\cD_n(a,b,c) \mid n\ge 3)$  of Definition~\ref{def:regular} and Figure~\ref{fig:RegWit} was introduced in~\cite{Brz13}  and studied   further in~\cite{BrSi16a}. It will be used as a component in all the classes of languages examined in this paper.
It was shown in~\cite{Brz13,BrSi16a} that this stream together with some dialects is most complex.  

\begin{definition}
\label{def:regular}
For $n\ge 3$, let $\cD_n=\cD_n(a,b,c)=(Q_n,\Sig,\delta_n, 0, \{n-1\})$, where 
$\Sig=\{a,b,c\}$, 
and $\delta_n$ is defined by 
$a\colon (0,\dots,n-1)$,
$b\colon(0,1)$,
$c\colon(n-1 \rightarrow 0)$.
\end{definition}

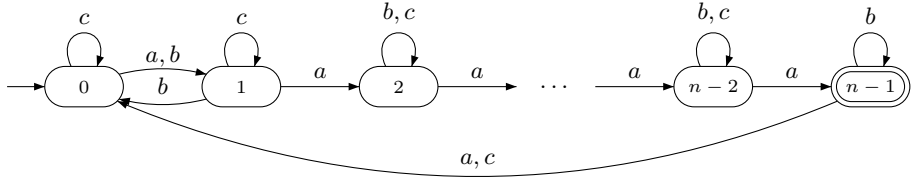
\begin{figure}[ht]
\unitlength 8.5pt
\begin{center}\begin{picture}(37,7)(0,3.5)
\gasset{Nh=1.8,Nw=3.5,Nmr=1.25,ELdist=0.4,loopdiam=1.5}
	{\scriptsize
\node(0)(1,7){0}\imark(0)
\node(1)(8,7){1}
\node(2)(15,7){2}
}
\node[Nframe=n](3dots)(22,7){$\dots$}
	{\scriptsize
\node(n-2)(29,7){$n-2$}
	}
	{\scriptsize
\node(n-1)(36,7){$n-1$}\rmark(n-1)
	}
\drawloop(0){$c$}
\drawedge[curvedepth= .8,ELdist=.1](0,1){$a,b$}
\drawedge[curvedepth= .8,ELdist=-1.2](1,0){$b$}
\drawedge(1,2){$a$}
\drawloop(2){$b,c$}
\drawedge(2,3dots){$a$}
\drawedge(3dots,n-2){$a$}
\drawloop(n-2){$b,c$}
\drawedge(n-2,n-1){$a$}
\drawedge[curvedepth= 4.0,ELdist=-1.0](n-1,0){$a,c$}
\drawloop(n-1){$b$}
\drawloop(1){$c$}
\end{picture}\end{center}
\caption{Minimal DFA of a most complex regular language.}
\label{fig:RegWit}
\end{figure}

\section{Left Ideals}

The following stream was  studied in~\cite{BrYe11} and also in~\cite{BrDa15,BDL15,BrSz14}.
This stream is most complex when the two alphabets are  the same in binary operations~\cite{BDL15}.
It is also most complex for unrestricted operations~\cite{BrSi16}.

\begin{definition}
\label{def:LWit}
For $n\ge 4$, let $\cD_n=\cD_n(a,b,c,d,e)=(Q_n,\Sig,\delta_n, 0, \{n-1\})$, where 
$\Sig=\{a,b,c,d,e\}$,
and $\delta_n$ is defined by  
transformations
$a\colon (1,\dots,n-1)$,
$b\colon(1,2)$,
${c\colon(n-1 \to 1)}$,
${d\colon(n-1\to 0)}$, 
and $e\colon (Q_n\to 1)$.
See Figure~\ref{fig:LWit}.
Let $L_n=L_n(a,b,c,d,e)$ be the language accepted by~$\cD_n$.
\end{definition}

\begin{figure}[h]
\unitlength 10pt
\begin{center}\begin{picture}(33,9)(-.5,3)
\gasset{Nh=2.4,Nw=2.7,Nmr=1.2,ELdist=0.3,loopdiam=1.2}
{\small
\node(0)(2,8){$0$}\imark(0)
\node(1)(7,8){$1$}
\node(2)(12,8){$2$}
\node(3)(17,8){$3$}
\node[Nframe=n](qdots)(22,8){$\dots$}

\node(n-2)(27,8){{\small $n-2$}}
\node(n-1)(31,8){{\small $n-1$}}\rmark(n-1)

\drawedge(0,1){$e$}
\drawedge[ELdist=.2](1,2){$a,b$}
\drawedge(2,3){$a$}
\drawedge(3,qdots){$a$}
\drawedge(qdots,n-2){$a$}
\drawedge(n-2,n-1){$a$}
\drawloop(0){$a,b,c,d$}
\drawloop(1){$c,d,e$}
\drawloop(2){$c,d$}
\drawloop(3){$b,c,d$}
\drawloop(n-2){$b,c,d$}
\drawloop(n-1){$b$}
\drawedge[curvedepth=-2.5](2,1){$b,e$}
\drawedge[curvedepth=-4.8](3,1){$e$}
\drawedge[curvedepth=2.1](n-2,1){$e$}
\drawedge[curvedepth=3.5](n-1,1){$a,c,e$}
\drawedge[curvedepth=5](n-1,0){$d$}
}
\end{picture}\end{center}
\caption{Minimal DFA 
of a most complex left ideal $L_n(a,b,c,d,e)$.
}
\label{fig:LWit}
\end{figure}
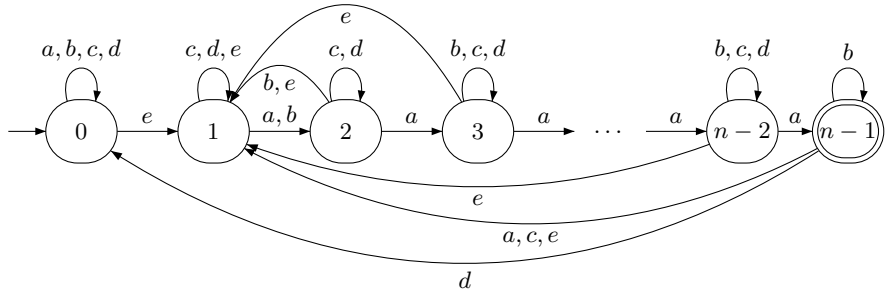

\begin{theorem}[Most Complex Left Ideals~\cite{BDL15,BrSi16}]
\label{thm:leftideals}
For each $n\ge 4$, the DFA of Definition~\ref{def:LWit} is minimal. 
The stream $(L_n(a,b,c,d,e) \mid n \ge 4)$  with some dialect streams
is most complex in the class of regular left ideals.
\be
\item
The syntactic semigroup of $L_n(a,b,c,d,e)$ has cardinality $n^{n-1}+n-1$.  
\item
Each quotient of $L_n(a,-,-,d,e)$ has complexity $n$.
\item
The reverse of $L_n(a,-,c,d,e)$ has complexity $2^{n-1}+1$, and $L_n(a,-,c,d,e)$ has $2^{n-1}+1$ atoms. 
\item
For each atom $A_S$ of $L_n(a,b,c,d,e)$, the complexity $\kappa(A_S)$ satisfies:
\begin{equation*}
	\kappa(A_S) =
	\begin{cases}
		 n, 			& \text{if $S=Q_n$;}\\
		2^{n-1},		& \text{if $S=\emp$;}\\
		1 + \sum_{x=1}^{|S|}\sum_{y=1}^{n-|S|}\binom{n-1}{x}\binom{n-x-1}{y-1},
		 			& \text{otherwise.}
		\end{cases}
\end{equation*}

\item
The star of $L_n(a,-,-,-,e)$ has complexity $n+1$.
\item
	\be
	\item
	Restricted product:
	$\kappa(L'_m(a,-,-,-,e) L_n(a,-,-,-,e)) = m+n-1$.
	\item
	Unrestricted product:
	$\kappa(L'_m(a,b,-,d,e) L_n(a,d,-,c,e)) = mn+m+n$.
	\ee
\item
	\be
	\item
	Restricted complexity:
	$\kappa(L'_m(a,-,c,-,e)\circ L_n(a,-,e,-,c)) = mn$.
	\item
	Unrestricted complexity:
	$\kappa(L'_m(a,b,-,d,e) \circ L_n(a,d,-,c,e) =
	(m+1)(n+1)$ if $\circ\in \{\cup,\oplus\})$, 
 $mn+m$ if $\circ=\setminus$, and $mn$ if $\circ=\cap$. 

	\ee
In both cases these bounds are the same as those for regular languages.

\ee
\end{theorem}

We now define a new left-ideal witness similar to the witness in Definition~\ref{def:LWit}.
\begin{definition}
\label{def:LNew}
For $n\ge 4$, let $\cE_n=\cE_n(a,b,c,d,e)=(Q_n,\Sig,\delta_n, 0, \{1,\dots,n-1\})$, where 
$\Sig$ and the transformations induced by its letters are as in $\cD_n$ of 
Definition~\ref{def:LWit}.
Let $M_n=M_n(a,b,c,d,e)$ be the language accepted by~$\cE_n$. 
\end{definition}

\begin{theorem}[Nearly Most Complex Left Ideals]
\label{thm:leftideals2}
For each $n\ge 4$, the DFA of Definition~\ref{def:LNew} is minimal and its 
language $M_n(a,b,c,d,e)$ is a left ideal of complexity $n$.
The stream $(M_n(a,b,c,d,e) \mid n \ge 4)$  with some dialect streams
meets all the complexity bounds for left ideals, except those for product.
\end{theorem}

\begin{proof}
It is easily verified that $\cE_n(a,-,-,d,e)$ is minimal; hence $M_n(a,b,c,d,e)$ has complexity $n$. $M_n$ is a left ideal because, for each letter $\ell$ of $\Sig$, and each word $w \in \Sig^*$,  $w\in M_n$ implies $\ell w\in M_n$.
We prove all the claims of Theorem~\ref{thm:leftideals} except the claims in  Item 6.

\be
\item
{\bf Semigroup} The  transition semigroup is independent of the set of final states; hence it has the size of the DFA of the most complex left ideal.

\item 
{\bf Quotients} 
Obvious.

\item {\bf Reversal} 
The upper bound of $2^{n-1}+1$ was proved in~\cite{BJL13}, and it was shown in~\cite{BrTa14} that the number of atoms is the same as the complexity of the reverse.
Applying the standard NFA construction for reversal, we reverse every transition in DFA $\cE_n$ and interchange the final and initial states, yielding the NFA in Figure~\ref{fig:sclosedreversal}, where the initial states (unmarked) are $Q_n\setminus \{0\}$.

We perform the subset construction.
Set $Q_n\setminus \{0\}$ is initial.
From $\{q_1, \dots,  q_{k}\}$, $1\le q_1 \le q_k$,  we delete $q_i$, 
$q_1\le q_i \le q_k\le n-1$, by applying $a^{q_i} d a^{n-1-q_i}$. 
Thus all $2^{n-1}$ subsets of $Q_n\setminus \{0\}$ can be reached,
 and $Q_n$ is reached from the initial state $\{1\}$ by $e$.
For any distinct $S, T \subseteq Q_n$ with $q \in S \setminus T$,
either $q=0$, in which case $S$ is final and $T$ is non-final,
or $Sa^{q-1}e = Q_n$ and $Ta^{q-1}e = \emptyset$.
Hence all $2^{n-1}+1$ states are pairwise distinguishable.

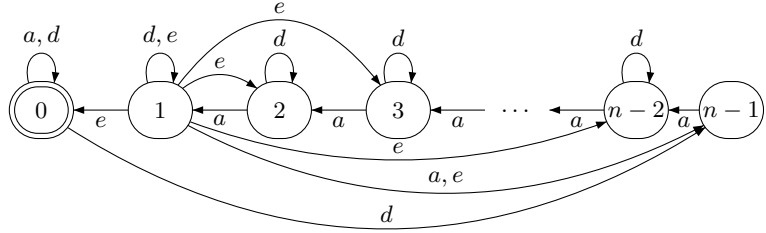
\begin{figure}[h]
\unitlength 9pt
\begin{center}\begin{picture}(33,9)(0,3.5)
\gasset{Nh=2.4,Nw=2.7,Nmr=1.2,ELdist=0.3,loopdiam=1.2}
{\small
\node(0)(2,8){$0$}\rmark(0)
\node(1)(7,8){$1$}
\node(2)(12,8){$2$}
\node(3)(17,8){$3$}
\node[Nframe=n](qdots)(22,8){$\dots$}

\node(n-2)(27,8){{\small $n-2$}}
\node(n-1)(31,8){{\small $n-1$}}
\drawedge(1,0){$e$}
\drawedge[ELdist=.2](2,1){$a$}
\drawedge(3,2){$a$}
\drawedge(qdots,3){$a$}
\drawedge(n-2,qdots){$a$}
\drawedge(n-1,n-2){$a$}
\drawloop(0){$a,d$}
\drawloop(1){$d,e$}
\drawloop(2){$d$}
\drawloop(3){$d$}
\drawloop(n-2){$d$}
\drawedge[curvedepth=1.5](1,2){$e$}
\drawedge[curvedepth=3.8](1,3){$e$}
\drawedge[curvedepth=-2.1](1,n-2){$e$}
\drawedge[curvedepth=-3.5](1,n-1){$a,e$}
\drawedge[curvedepth=-5](0,n-1){$d$}
}
\end{picture}\end{center}
\caption{NFA for reversal of $M_n(a,-,-,d,e)$.}
\label{fig:sclosedreversal}
\end{figure}

\item {\bf Atoms} 
The upper bounds   in Theorem~\ref{thm:leftideals} for left ideals were derived in~\cite{BrDa15}. The proof of~\cite{BrDa15} that these bounds are met applies also to our witness $M_n$.

\item {\bf Star}
The upper bound $n+1$ was proved in~\cite{BJL13}.
To construct an NFA recognizing $(M_n(a,-,-,d,e))^*$ we add a new initial state $0'$ which is also final and has the same transitions as the former initial state $0$. We then 
add an $\eps$-transition from each final state of $\cE(a,-,-,d,e)$ to the initial state $0'$.
The language recognized by the new NFA $\cN$ is 
$(M_n(a,-,-,d,e))^*$. The final state $\{0'\}$ in the subset construction for $\cN$ is distinguishable from every other final state, since it rejects $a$, whereas other final states accept it. 

\item 
{\bf Product} Not applicable.

\item {\bf Boolean Operations}
\be
\item{Restricted complexity:}
The upper bound of $mn$ is the same as for regular languages.
We show that $M'_m(a,b,-,d,e) \circ M_n(a,e,-,d,b)$ has complexity $mn$.
In the standard construction for boolean operations, we consider the direct product of $\mathcal{E}'_m(a,b,-,d,e)$ and $\mathcal{E}_n(a,e,-,d,b)$.
The set of final states of the direct product varies depending on the operation $\circ \in \{\cup, \oplus, \setminus, \cap\}$.

\begin{figure}[ht]
\unitlength 8.5pt
\begin{center}\begin{picture}(35,21)(0,-2)
\gasset{Nh=2.6,Nw=2.6,Nmr=1.2,ELdist=0.3,loopdiam=1.2}
	{\scriptsize
\node(0'0)(2,15){$0',0$}\imark(0'0)
\node(1'0)(2,10){$1',0$}\rmark(1'0)
\node(2'0)(2,5){$2',0$}\rmark(2'0)
\node(3'0)(2,0){$3',0$}\rmark(3'0)

\node(0'1)(10,15){$0',1$}\rmark(0'1)
\node(1'1)(10,10){$1',1$}
\node(2'1)(10,5){$2',1$}
\node(3'1)(10,0){$3',1$}

\node(0'2)(18,15){$0',2$}\rmark(0'2)
\node(1'2)(18,10){$1',2$}
\node(2'2)(18,5){$2',2$}
\node(3'2)(18,0){$3',2$}

\node(0'3)(26,15){$0',3$}\rmark(0'3)
\node(1'3)(26,10){$1',3$}
\node(2'3)(26,5){$2',3$}
\node(3'3)(26,0){$3',3$}

\node(0'4)(34,15){$0',4$}\rmark(0'4)
\node(1'4)(34,10){$1',4$}
\node(2'4)(34,5){$2',4$}
\node(3'4)(34,0){$3',4$}
}

\drawedge(0'0,0'1){$b$}
\drawedge(0'0,1'0){$e$}

\drawedge(1'0,2'0){$a$}
\drawedge(2'0,3'0){$a$}
\drawedge[curvedepth=2, ELpos=65](3'0,1'0){$a$}
\drawedge[curvedepth=3](3'0,0'0){$d$}

\drawedge[ELpos=35](1'0,2'1){$b$}
\drawedge[ELpos=35](2'0,1'1){$b$}
\drawedge(3'0,3'1){$b$}

\drawedge(0'1,0'2){$a$}
\drawedge(0'2,0'3){$a$}
\drawedge(0'3,0'4){$a$}
\drawedge[curvedepth=-2, ELside=r, ELpos=65](0'4,0'1){$a$}
\drawedge[curvedepth=-3, ELside=r](0'4,0'0){$d$}

\drawedge[ELpos=35, ELside=r](0'1,1'2){$e$}
\drawedge[ELpos=35](0'2,1'1){$e$}
\drawedge(0'3,1'3){$e$}
\drawedge(0'4,1'4){$e$}

\drawedge(1'1,2'2){$a$}
\drawedge(2'2,3'3){$a$}
\drawedge(3'3,1'4){$a$}
\drawedge(3'1,1'2){$a$}
\drawedge(1'2,2'3){$a$}
\drawedge(2'3,3'4){$a$}
\drawedge(2'1,3'2){$a$}
\drawedge(3'2,1'3){$a$}
\drawedge(1'3,2'4){$a$}
\end{picture}\end{center}
\caption{Direct product for $M'_4(a,b,-,d,e)\oplus M_5(a,e,-,d,b)$ shown partially.}
\label{fig:idealcross1}
\end{figure}
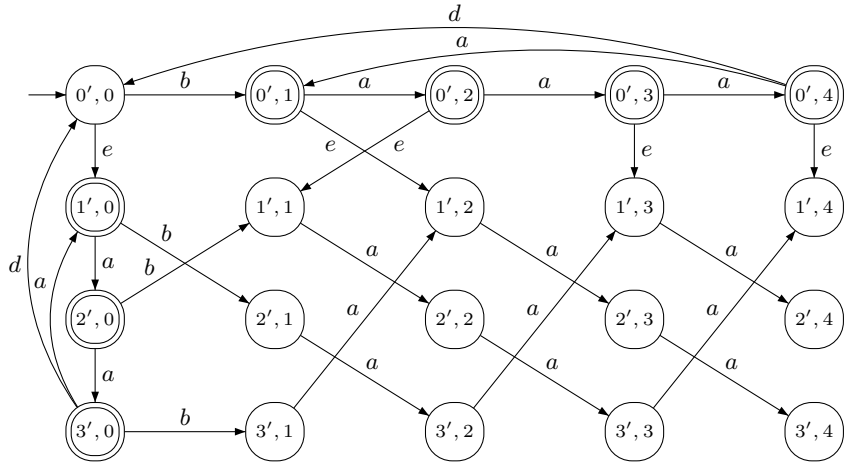

We first check that all $mn$ states are reachable in the direct product.
State $(0',0)$ is initial and $(p',0)$, $1\le p\le m-1$, is reached by $ea^{p-1}$.
If $1 \le q \le p$ then $(p', q)$ is reached from $((p-q+1)', 0)$ by $b^2a^{q-1}$.	
Similarly $(0', q)$ is reached by $ba^{q-1}$, and if $1 \le p \le q$ then $(p', q)$ is reached from $(0', q-p+1)$ by $e^2a^{p-1}$.
Hence all $mn$ states are reachable.

Let $R = \{ (0', q) \mid q \in Q_n\setminus\{0\}\}$, $C = \{ (p',0) \mid p' \in Q'_m\setminus\{0'\}\}$, and $S = \{(p',q) \mid p' \in Q'_m\setminus\{0\}, q \in Q_n\setminus\{0\}\}$.
States of $S$ are pairwise distinguished with respect to the set $R \cup C \cup \{(0',0)\}$ by words in $a^*bd$ if they differ in the first coordinate, or by words in $a^*ed$ if they differ in the second coordinate. We consider each operation separately to show that the states of $R \cup C \cup \{(0',0)\}$ are pairwise distinguishable and distinguishable from the states of $S$, with respect to the final states.

{\bf Union}
All states are final except $(0',0)$.
States of $R$ are distinguished by words in $a^*d$, as are states of $C$.
States of $R$ are distinguishable from those of $C \cup S$ because every state of $C \cup S$ accepts either $ba^{n-2}d$ or $aba^{n-2}d$, while states of $R$ are sent to $(0',0)$ by any word in $a^*ba^{n-2}d$. States of $C$ are similarly distinguishable from those of $R \cup S$.

{\bf Symmetric Difference}
The final states are those in $R \cup C$.
The argument is the same as union, except $(0',0)$ is distinguished from states of $S$ by $e$.

{\bf Difference}
The final states are those in $C$.
States of $C$ are distinguished by words in $a^*d$, and states of $R$ are distinguished by words in $a^*de$.
State $(0',0)$ is distinguished from states of $R$ by $e$.
States of $R \cup \{(0',0)\}$ are distinguishable from states of $S$ because every state of $S$ accepts a word in $\{a,b,d\}^*$, while those of $C \cup \{(0',0)\}$ accept only words with $e$.

{\bf Intersection}
The final states are those in $S$.
States of $R$ are distinguished by words in $a^*d$, as are states of $C$.
States of $R$ are distinguished from states of $C$ by $e$.
State $(0',0)$ is distinguished from states of $R$ by $e$ and from states of $C$ by $b$.

Hence $M'_m(a,b,-,d,e) \circ M_n(a,e,-,d,b)$ has complexity $mn$ for each $\circ \in \{\cup, \oplus, \setminus, \cap\}$.\\

\item{Unrestricted complexity:}
To produce a DFA recognizing $M'_m(a,b,c,d,e) \circ M_n(a,e,f,d,b)$, where $\circ$ is a boolean operation, we first add an empty state $\emptyset'$ to $\mathcal{E}'_m(a,b,c,d,e)$, and send all the transitions from any state of $Q'_m$ under $f$ to $\emptyset'$.
Similarly add an empty state $\emptyset$ to $\mathcal{E}_n(a,e,f,d,b)$ and send all the transitions from any state of $Q_n$ under $c$ to $\emptyset$.
Now the DFAs are over the combined alphabet $\{a,b,c,d,e,f\}$ and we take the direct product as before; the direct product for union is illustrated in Figure~\ref{fig:idealcross2}.
	
\begin{figure}[ht]
\unitlength 8.5pt
\begin{center}\begin{picture}(35,28)(0,-9)
\gasset{Nh=2.6,Nw=2.6,Nmr=1.2,ELdist=0.3,loopdiam=1.2}
	{\scriptsize
\node(0'0)(2,15){$0',0$}\imark(0'0)
\node(1'0)(2,10){$1',0$}\rmark(1'0)
\node(2'0)(2,5){$2',0$}\rmark(2'0)
\node(3'0)(2,0){$3',0$}\rmark(3'0)
\node(4'0)(2,-5){$\emp',0$}

\node(0'1)(10,15){$0',1$}\rmark(0'1)
\node(1'1)(10,10){$1',1$}\rmark(1'1)
\node(2'1)(10,5){$2',1$}\rmark(2'1)
\node(3'1)(10,0){$3',1$}\rmark(3'1)
\node(4'1)(10,-5){$\emp',1$}\rmark(4'1)

\node(0'2)(18,15){$0',2$}\rmark(0'2)
\node(1'2)(18,10){$1',2$}\rmark(1'2)
\node(2'2)(18,5){$2',2$}\rmark(2'2)
\node(3'2)(18,0){$3',2$}\rmark(3'2)
\node(4'2)(18,-5){$\emp',2$}\rmark(4'2)

\node(0'3)(26,15){$0',3$}\rmark(0'3)
\node(1'3)(26,10){$1',3$}\rmark(1'3)
\node(2'3)(26,5){$2',3$}\rmark(2'3)
\node(3'3)(26,0){$3',3$}\rmark(3'3)
\node(4'3)(26,-5){$\emp',3$}\rmark(4'3)

\node(0'4)(34,15){$0',\emp$}
\node(1'4)(34,10){$1',\emp$}\rmark(1'4)
\node(2'4)(34,5){$2',\emp$}\rmark(2'4)
\node(3'4)(34,0){$3',\emp$}\rmark(3'4)
\node(4'4)(34,-5){$\emp',\emp$}
	}

\drawedge(0'0,0'1){$b$}
\drawedge(0'0,1'0){$e$}

\drawedge(1'0,2'0){$a$}
\drawedge(2'0,3'0){$a$}
\drawedge[curvedepth=2, ELpos=65](3'0,1'0){$a$}
\drawedge[curvedepth=3](3'0,0'0){$d$}

\drawedge[ELpos=35](1'0,2'1){$b$}
\drawedge[ELpos=35, ELside=r](2'0,1'1){$b$}
\drawedge(3'0,3'1){$b$}

\drawedge(0'1,0'2){$a$}
\drawedge(0'2,0'3){$a$}
\drawedge[curvedepth=-2, ELside=r, ELpos=65](0'3,0'1){$a$}
\drawedge[curvedepth=-3, ELside=r](0'3,0'0){$d$}

\drawedge[ELpos=35, ELside=r](0'1,1'2){$e$}
\drawedge[ELpos=35](0'2,1'1){$e$}
\drawedge(0'3,1'3){$e$}

\drawedge(1'1,2'2){$a$}
\drawedge(2'2,3'3){$a$}
\drawedge(1'2,2'3){$a$}
\drawedge[curvedepth=-0.6](2'3,3'1){$a$}
\drawedge(2'1,3'2){$a$}
\drawedge[curvedepth=0.8](3'2,1'3){$a$}

\drawedge(0'3,0'4){$c$}
\drawedge(1'3,1'4){$c$}
\drawedge(2'3,2'4){$c$}
\drawedge(4'3,4'4){$c$}

\drawedge(3'0,4'0){$f$}
\drawedge(3'1,4'1){$f$}
\drawedge(3'2,4'2){$f$}
\drawedge(3'4,4'4){$f$}

\drawedge(4'0,4'1){$b$}
\drawedge(4'1,4'2){$a$}
\drawedge(4'2,4'3){$a$}
\drawedge[curvedepth=2, ELpos=65](4'3,4'1){$a$}
\drawedge[curvedepth=3](4'3,4'0){$d$}

\drawedge(0'4,1'4){$e$}
\drawedge(1'4,2'4){$a$}
\drawedge(2'4,3'4){$a$}
\drawedge[curvedepth=-2, ELside=r, ELpos=65](3'4,1'4){$a$}
\drawedge[curvedepth=-3, ELside=r](3'4,0'4){$d$}

\end{picture}\end{center}
\caption{Direct product for $M'_4(a,b,c,d,e)\cup M_4(a,e,f,d,b)$ shown partially.}
\label{fig:idealcross2}
\end{figure}
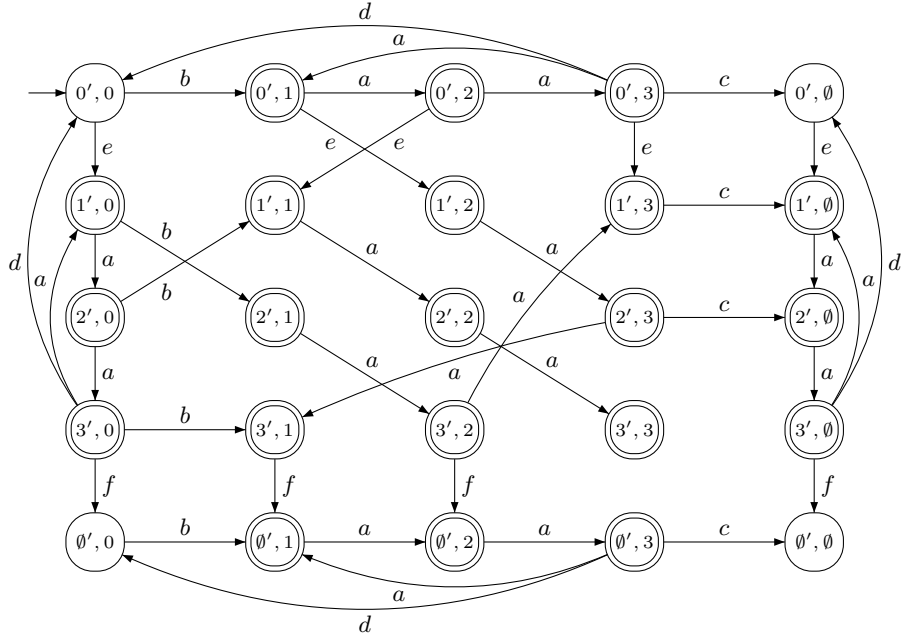
	
By the restricted case all the states of $Q'_m \times Q_n$ are reachable and distinguishable using words in $\{a,b,d,e\}^*$.
Let $R_{\emptyset'} = \{(\emptyset', q) \mid q \in Q_n\}$ and $C_{\emptyset} = \{(p', \emptyset) \mid p' \in Q'_m\}$.
States of $R_{\emptyset'} \cup C_{\emptyset} \cup \{ (\emp',\emp)\}$ are easily seen to be reachable using $c$ and $f$ in addition to $a$, $b$, $d$, and $e$.
We check that the states of $R_{\emptyset'} \cup C_{\emptyset} \cup \{ (\emp',\emp) \}$ are pairwise distinguishable and distinguishable from the states of $Q'_m \times Q_n$.

{\bf Union}
The final states of $R_{\emptyset'}$ are distinguished by words in $a^*d$, and those  of $C_{\emptyset}$ are similarly distinguishable.
All states except $\{ (\emp',\emp)\}$ are non-empty since each accepts a word in $\{a,b,d,e\}^*$.
States of $R_{\emptyset'} \cup \{(\emptyset', \emptyset)\}$ are distinguishable from all other states since every other state accepts $ce$.
Similarly, states of $C_{\emptyset} \cup \{(\emptyset', \emptyset)\}$ are distinguishable from all other states since every other state accepts $fb$.
Hence all $(m+1)(n+1)$ states are pairwise distinguishable.

{\bf Symmetric Difference}
Same as union.

{\bf Difference}
States of $R_{\emptyset'} \cup \{ (\emp',\emp)\}$ are empty and therefore equivalent.
However, since the alphabet of $M'_m(a,b,c,d,e) \setminus M_n(a,e,f,d,b)$ is $\{a,b,c,d,e\}$ we can omit $f$ and delete the states of $R_{\emptyset'} \cup \{(\emp',\emp)\}$, and be left with a DFA over $\{a,b,c,d,e\}$ that recognizes $M'_m(a,b,c,d,e) \setminus M_n(a,e,f,d,b)$.
States of $C_{\emptyset}$ are distinguished by words in $a^*d$, and
states of $Q'_m \times Q_n$ are distinguished from states of $C_{\emptyset}$ by $be$.
Hence the $mn+m$ remaining states are pairwise distinguishable.

{\bf Intersection} States of $R_{\emptyset'} \cup C_{\emptyset} \cup \{ (\emp',\emp)\}$ are empty and therefore equivalent.
However, since the alphabet of $M'_m(a,b,c,d,e) \cap M_n(a,e,f,d,b)$ is $\{a,b,d,e\}$, we can omit $c$ and $f$ and delete the states of $R_{\emptyset'} \cup C_{\emptyset} \cup \{(\emp',\emp)\}$, and be left with a DFA over $\{a,b,d,e\}$ that recognizes $M'_m(a,b,c,d,e) \cap M_n(a,e,f,d,b)$. By the restricted case, all $mn$ states are pairwise distinguishable.
\qed
\ee
\ee
\end{proof}

\section{Suffix-Closed Languages}
\label{sec:sclosed}

The  complexity of suffix-closed languages was studied in~\cite{BJZ14} in the restricted case, and the syntactic semigroup of these languages, in~\cite{BrSz14,BSY15,BrYe11}; however, most complex suffix-closed languages have not been examined.
\begin{definition}
\label{def:sclosed}
For $n\ge 4$, let $\cD_n=\cD_n(a,b,c,d,e)=(Q_n,\Sig,\delta_n, 0, \{0\})$, where 
$\Sig=\{a,b,c,d,e\}$,
and $\delta_n$ is defined by  transformations
$a\colon (1,\dots,n-1)$,
$b\colon(1,2)$,
${c\colon(n-1 \to 1)}$,
${d\colon(n-1\to 0)}$, 
$e\colon (Q_n\to 1)$.
Let $L_n=L_n(a,b,c,d,e)$ be the language accepted by~$\cD_n$; this language is the complement of the left ideal of Definition~\ref{def:LNew}.
The structure of  $\cD_n(a,b,c,d,e)$ is shown in Figure~\ref{fig:sclosed}. 
\end{definition}

\begin{figure}[h]
\unitlength 9.5pt
\begin{center}\begin{picture}(33,9)(-.5,3)
\gasset{Nh=2.4,Nw=2.7,Nmr=1.2,ELdist=0.3,loopdiam=1.2}
{\small
\node(0)(2,8){$0$}\imark(0)\rmark(0)
\node(1)(7,8){$1$}
\node(2)(12,8){$2$}
\node(3)(17,8){$3$}
\node[Nframe=n](qdots)(22,8){$\dots$}

\node(n-2)(27,8){{\small $n-2$}}
\node(n-1)(31,8){{\small $n-1$}}
\drawedge(0,1){$e$}
\drawedge[ELdist=.2](1,2){$a,b$}
\drawedge(2,3){$a$}
\drawedge(3,qdots){$a$}
\drawedge(qdots,n-2){$a$}
\drawedge(n-2,n-1){$a$}
\drawloop(0){$a,b,c,d$}
\drawloop(1){$c,d,e$}
\drawloop(2){$c,d$}
\drawloop(3){$b,c,d$}
\drawloop(n-2){$b,c,d$}
\drawloop(n-1){$b$}
\drawedge[curvedepth=-2.5](2,1){$b,e$}
\drawedge[curvedepth=-4.8](3,1){$e$}
\drawedge[curvedepth=2.1](n-2,1){$e$}
\drawedge[curvedepth=3.5](n-1,1){$a,c,e$}
\drawedge[curvedepth=5](n-1,0){$d$}
}
\end{picture}\end{center}
\caption{Minimal DFA $\cD_n(a,b,c,d,e)$  of Definition~\ref{def:sclosed}.}
\label{fig:sclosed}
\end{figure}
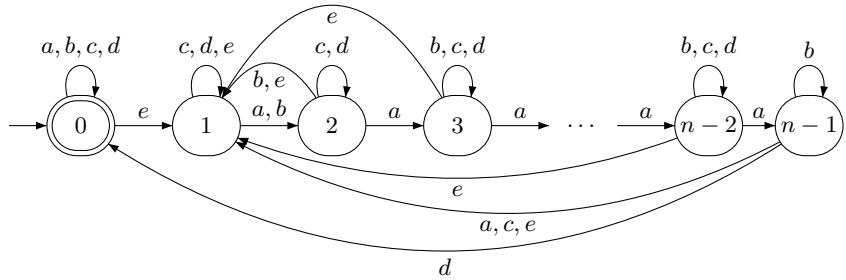

\begin{theorem}[Most Complex Suffix-Closed Languages]
\label{thm:sclosedmain}
For each $n\ge 4$, the DFA of Definition~\ref{def:sclosed} is minimal and its 
language $L_n(a,b,c,d,e)$ is suffix-closed and has complexity $n$.
The stream $(L_m(a,b,c,d,e) \mid m \ge 4)$  with some dialect streams
is most complex in the class of suffix-closed languages.
\begin{enumerate}
\item
The syntactic semigroup of $L_n(a,b,c,d,e)$ has cardinality $n^{n-1}+n-1$.  Moreover, fewer than five inputs do not suffice to meet this bound.
\item
All quotients of $L_n(a,-,-,d,e)$ have complexity $n$.
\item
The reverse of $L_n(a,-,-,d,e)$ has complexity $2^{n-1}+1$, and $L_n(a,-,-,d,e)$ has $2^{n-1}+1$ atoms. 
\item
For each atom $A_S$ of $L_n(a,b,c,d,e)$, the complexity $\kappa(A_S)$ satisfies:
\begin{equation*}
	\kappa(A_S) =
	\begin{cases}
		 n, 			& \text{if $S=\emp$;}\\
		2^{n-1},		& \text{if $S=Q_n$;}\\
		1 + \sum_{x=1}^{|S|}\sum_{y=1}^{n-|S|}\binom{n-1}{y}\binom{n-y-1}{x-1},
		 			& \text{if $\{0\} \subseteq S \subsetneq Q_n$.}
		\end{cases}
\end{equation*}

\item
The star of $L_n(a,-,-,d,e)$ has complexity $n$.
\item
\be
	\item
	Restricted complexity:
	$\kappa(L'_m(a,b,-,d,e) L_n(a,e,-,d,b)) = mn-n+1$.
	\item
	Unrestricted complexity:
	$\kappa(L'_m(a,b,c,d,e) L_n(a,e,f,d,b)) = mn+m+1$.
	\ee
\item
	\be
	\item
	Restricted complexity:
	$\kappa( L'_m(a,b,-,d,e) \circ L_n(a,e,-,d,b) ) = mn$ for $\circ \in \{\cup, \oplus, \cap, \setminus\}$.
	\item
	Unrestricted complexity:
	$\kappa( L'_m(a,b,c,d,e) \circ L_n(a,e,f,d,b) ) =
	( m+1)(n+1)$ if $\circ\in \{\cup,\oplus\}$, 
it is  $mn+m$ if $\circ=\setminus$, and $mn$ if $\circ=\cap$. 
	\ee
\end{enumerate}
\end{theorem}
\begin{proof}
DFA $\cD_n(a,-,-,d,e)$ is minimal and $L_n(a,b,c,d,e)$ is suffix-closed since its complement is a left ideal. 
\be
\item
{\bf Semigroup}
The transition semigroup is independent of the set of final states; hence its size is the same as that of the transition semigroup of the DFA $\cE_n$ of the left ideal $M_n$.
\item
{\bf Quotients} Obvious.
\item {\bf Reversal} 
This follows from the results for $M_n$, since complementation commutes with reversal.

\item {\bf Atoms} 
We first establish an upper bound on the complexity of the atoms, using the corresponding bounds for left ideals.
Let $L$  be a suffix-closed language with quotients $K_0, \dots, K_{n-1}$; then $\overline{L}$ is a left ideal with quotients $\overline{K_0}, \dots, \overline{K_{n-1}}$.
For $S \subseteq Q_n$, the atom of $L$ corresponding to $S$ is
$A_S = \bigcap_{i \in S} K_i \cap \bigcap_{i \in \overline{S}} \overline{K_i}$.
This can be rewritten as $\bigcap_{i \in \overline{S}} \overline{K_i} \cap \bigcap_{i \in \overline{\overline{S}}} \overline{\overline{K_i}}$, which is the atom of $\overline{L}$ corresponding to $\overline{S}$; hence the sets of atoms of $L$ and $\overline{L}$ are the same.
The upper bounds now follow from those for left ideals as given in Theorem~\ref{thm:leftideals}, which were derived in~\cite{BrDa15}.
\item {\bf Star}
The upper bound $n$ was proved in~\cite{BJZ14}.
To construct an NFA recognizing $(L_{ n}(a,-,-,d,e))^*$ we add an $\eps$-transition from the final state of $\cD_{ n}(a,-,-,d,e)$ to the initial state $0$;
however in this case the $\eps$-transition is a loop at 0, which does not affect the language recognized by the automaton.
Thus $(L_{ n}(a,-,-,d,e))^* = L_n(a,-,-,d,e)$ and its complexity is $n$.

\item	{\bf Product}
	\be
	\item
	Restricted complexity:
	The upper bound $mn-n+1$ was derived in~\cite{BJZ14}.
	The NFA for the product $L'_m(a,b,-,d,e)L_n(a,e,-,d,b)$ is shown in Figure~\ref{fig:sclosedprod1} for $m=n=4$.

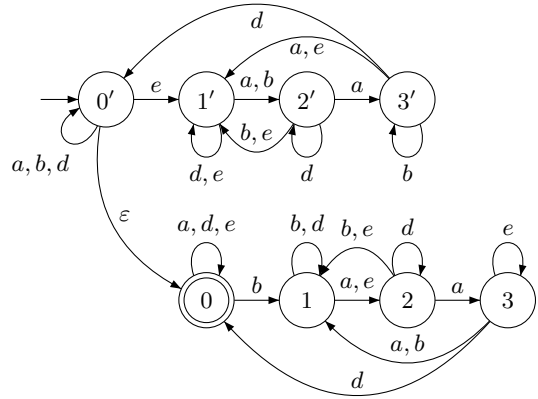
\begin{figure}[h]
\unitlength 9.5pt
\begin{center}\begin{picture}(31,16)(-2,0)
\gasset{Nh=2.2,Nw=2.2,Nmr=1.2,ELdist=0.3,loopdiam=1.2}
\node(0')(6,12){$0'$}\imark(0')
\node(1')(10,12){$1'$}
\node(2')(14,12){$2'$}
\node(3')(18,12){$3'$}

\drawedge(0',1'){$e$}
\drawedge(1',2'){$a,b$}
\drawedge(2',3'){$a$}
\drawloop[loopangle=225](0'){$a,b,d$}
\drawloop[loopangle=270](1'){$d,e$}
\drawloop[loopangle=270](2'){$d$}
\drawloop[loopangle=270](3'){$b$}
\drawedge[curvedepth=2.1, ELside=r](2',1'){$b,e$}
\drawedge[curvedepth=-2.5](3',1'){$a,e$}
\drawedge[curvedepth=-3.8](3',0'){$d$}

\node(0)(10,4){$0$}\rmark(0)
\node(1)(14,4){$1$}
\node(2)(18,4){$2$}
\node(3)(22,4){$3$}

\drawedge(0,1){$b$}
\drawedge(1,2){$a,e$}
\drawedge(2,3){$a$}
\drawloop(0){$a,d,e$}
\drawloop(1){$b,d$}
\drawloop(2){$d$}
\drawloop(3){$e$}
\drawedge[curvedepth=-2.1, ELside=r](2,1){$b,e$}
\drawedge[curvedepth=2.5, ELside=r](3,1){$a,b$}
\drawedge[curvedepth=3.8, ELside=r](3,0){$d$}

\drawedge[curvedepth=-2](0',0){$\eps$}
\end{picture}\end{center}
\caption{NFA for product of suffix-closed languages $L'_m(a,b,-,d,e)$ and $L_n(a,e,-,d,b)$.}
\label{fig:sclosedprod1}
\end{figure}

As $b \colon (1',2')(Q_n \to 1)$ is the only letter which does not fix $0$ and since $b$ maps $Q_n$ to 1,
the reachable sets in the subset construction are of the form $\{p', q\}$ or $\{p',0,q\}$ for $p' \in Q'_m$ and $q \in Q_n$.
However we cannot reach sets $\{0',q\}$ where $q \not= 0$, due to the $\eps$-transition from $0'$ to $0$.
Furthermore, the states $\{ \{p',0,q\} \mid q \in Q_n\}$ are equivalent as any word that maps $q$ to $0$ also fixes $0$.
Hence we consider only sets $\{p',q\}$ for $p' \in Q'_m \setminus \{0'\}$ and $q \in Q_n$, and the initial state $\{0',0\}$; note that there are $mn-n+1$ such sets.

For $p'\in Q'_m\setminus \{0'\}$ set $\{p',0\}$ is reached by $ea^{p-1}$, and $\{p',q\}$ is reached from $\{r',0\}$ by $ba^{q-1}$ for some $r' \in Q'_m \setminus \{0'\}$, since $ba^{q-1}$ induces a permutation on $Q'_m$.
Thus all $mn-n+1$ states are reachable.

Among these states, only $\{0',0\}$ is final. Non-final states $\{p'_1,q_1\}$ and $\{p'_2,q_2\}$ are distinguished by a word in $ea^*d$ if $q_1 \not= q_2$, or by a word in $ba^*d$ if $p'_1 \not= p'_2$. Thus they are pairwise distinguishable and the product has complexity $mn-n+1$.

	\item
	Unrestricted complexity:
The NFA for $L'_m(a,b,c,d,e)L_n(a,e,f,d,b)$ is the same as Figure~\ref{fig:sclosedprod1} except for the additional transformations $c \colon ((m-1)' \to 1')(Q_n \to \emptyset)$ and $f \colon (n-1 \to 1)(Q'_m \to \emptyset)$.
In addition to the $mn-n+1$ reachable and distinguishable states of the restricted case, $c$ and $f$ allow us to reach $\{p'\}$ for $p' \in Q'_m \setminus \{0'\}$, $\{q\}$ for $q \in Q_n$, and $\emptyset$.
State $\{p'\}$ is reached from the initial state by $eca^{p-1}$, $\{0\}$ is reached by $f$, and $\{q\}$ is reached by $fba^{q-1}$. The empty set is reached by $fc$.

The original $mn-n+1$ states are pairwise distinguishable as before.
States $\{p'\}$ for $p' \in Q'_m \setminus \{0'\}$ are pairwise distinguishable by words in $a^*d$,
as are states $\{q\}$ for $q \in Q_n$.
All states other than $\emptyset$ are non-empty, since they all accept $ea^{m-2}d$ or $ba^{n-2}d$.
All states $\{p'\}$ are distinguishable from states containing elements of $Q_n$, since $\{p'\}f = \emptyset$ while $\{q\}f \not= \emptyset$ for all $q \in Q_n$.
Similarly, all states $\{q\}$ are distinguishable from states containing elements of $Q'_m$.
Thus, all $mn+m+1$ states are pairwise distinguishable.
\ee
\smallskip

\item {\bf Boolean Operations}
\be
\item{Restricted complexity:}
Since $L'_m(a,b,-,d,e)$ ($L_n(a,e,-,d,b)$) is the complement of the left ideal $M'_m(a,b,-,d,e)$ ($M_n(a,e,-,d,b)$) of Definition~\ref{def:LNew} and they share a common alphabet $\{a,b,d,e\}$, by DeMorgan's laws we have
$\kappa(L'_m \cup L_n) = \kappa(M'_m \cap M_n)$,
$\kappa(L'_m \oplus L_n) = \kappa(M'_m \oplus M_n)$,
$\kappa(L'_m \setminus L_n) = \kappa(M_n \setminus M'_m)$,
and $\kappa(L'_m \cap L_n) = \kappa(M'_m \cup M_n)$.
Thus, by Theorem~\ref{thm:leftideals2}, all boolean operations have complexity $mn$.\\

\item{Unrestricted complexity:}
Following the example set in Theorem~\ref{thm:leftideals2}, we take the direct product for $L'_m(a,b,c,d,e) \circ L_n(a,e,f,d,b)$, as illustrated in Figure~\ref{fig:sclosedcross2} for $\circ = \cup$.
	
\begin{figure}[ht]
\unitlength 8.5pt
\begin{center}\begin{picture}(35,28)(0,-9)
\gasset{Nh=2.6,Nw=2.6,Nmr=1.2,ELdist=0.3,loopdiam=1.2}
	{\scriptsize
\node(0'0)(2,15){$0',0$}\imark(0'0)\rmark(0'0)
\node(1'0)(2,10){$1',0$}\rmark(1'0)
\node(2'0)(2,5){$2',0$}\rmark(2'0)
\node(3'0)(2,0){$3',0$}\rmark(3'0)
\node(4'0)(2,-5){$\emp',0$}\rmark(4'0)

\node(0'1)(10,15){$0',1$}\rmark(0'1)
\node(1'1)(10,10){$1',1$}
\node(2'1)(10,5){$2',1$}
\node(3'1)(10,0){$3',1$}
\node(4'1)(10,-5){$\emp',1$}

\node(0'2)(18,15){$0',2$}\rmark(0'2)
\node(1'2)(18,10){$1',2$}
\node(2'2)(18,5){$2',2$}
\node(3'2)(18,0){$3',2$}
\node(4'2)(18,-5){$\emp',2$}

\node(0'3)(26,15){$0',3$}\rmark(0'3)
\node(1'3)(26,10){$1',3$}
\node(2'3)(26,5){$2',3$}
\node(3'3)(26,0){$3',3$}
\node(4'3)(26,-5){$\emp',3$}

\node(0'4)(34,15){$0',\emp$}\rmark(0'4)
\node(1'4)(34,10){$1',\emp$}
\node(2'4)(34,5){$2',\emp$}
\node(3'4)(34,0){$3',\emp$}
\node(4'4)(34,-5){$\emp',\emp$}
	}

\drawedge(0'0,0'1){$b$}
\drawedge(0'0,1'0){$e$}

\drawedge(1'0,2'0){$a$}
\drawedge(2'0,3'0){$a$}
\drawedge[curvedepth=2, ELpos=65](3'0,1'0){$a$}
\drawedge[curvedepth=3](3'0,0'0){$d$}

\drawedge[ELpos=35](1'0,2'1){$b$}
\drawedge[ELpos=35, ELside=r](2'0,1'1){$b$}
\drawedge(3'0,3'1){$b$}

\drawedge(0'1,0'2){$a$}
\drawedge(0'2,0'3){$a$}
\drawedge[curvedepth=-2, ELside=r, ELpos=65](0'3,0'1){$a$}
\drawedge[curvedepth=-3, ELside=r](0'3,0'0){$d$}

\drawedge[ELpos=35, ELside=r](0'1,1'2){$e$}
\drawedge[ELpos=35](0'2,1'1){$e$}
\drawedge(0'3,1'3){$e$}

\drawedge(1'1,2'2){$a$}
\drawedge(2'2,3'3){$a$}
\drawedge(1'2,2'3){$a$}
\drawedge[curvedepth=-0.6](2'3,3'1){$a$}
\drawedge(2'1,3'2){$a$}
\drawedge[curvedepth=0.8](3'2,1'3){$a$}

\drawedge(0'3,0'4){$c$}
\drawedge(1'3,1'4){$c$}
\drawedge(2'3,2'4){$c$}
\drawedge(4'3,4'4){$c$}

\drawedge(3'0,4'0){$f$}
\drawedge(3'1,4'1){$f$}
\drawedge(3'2,4'2){$f$}
\drawedge(3'4,4'4){$f$}

\drawedge(4'0,4'1){$b$}
\drawedge(4'1,4'2){$a$}
\drawedge(4'2,4'3){$a$}
\drawedge[curvedepth=2, ELpos=65](4'3,4'1){$a$}
\drawedge[curvedepth=3](4'3,4'0){$d$}

\drawedge(0'4,1'4){$e$}
\drawedge(1'4,2'4){$a$}
\drawedge(2'4,3'4){$a$}
\drawedge[curvedepth=-2, ELside=r, ELpos=65](3'4,1'4){$a$}
\drawedge[curvedepth=-3, ELside=r](3'4,0'4){$d$}

\end{picture}\end{center}
\caption{Direct product for $L'_4(a,b,c,d,e)\cup L_4(a,e,f,d,b)$ shown partially.}
\label{fig:sclosedcross2}
\end{figure}
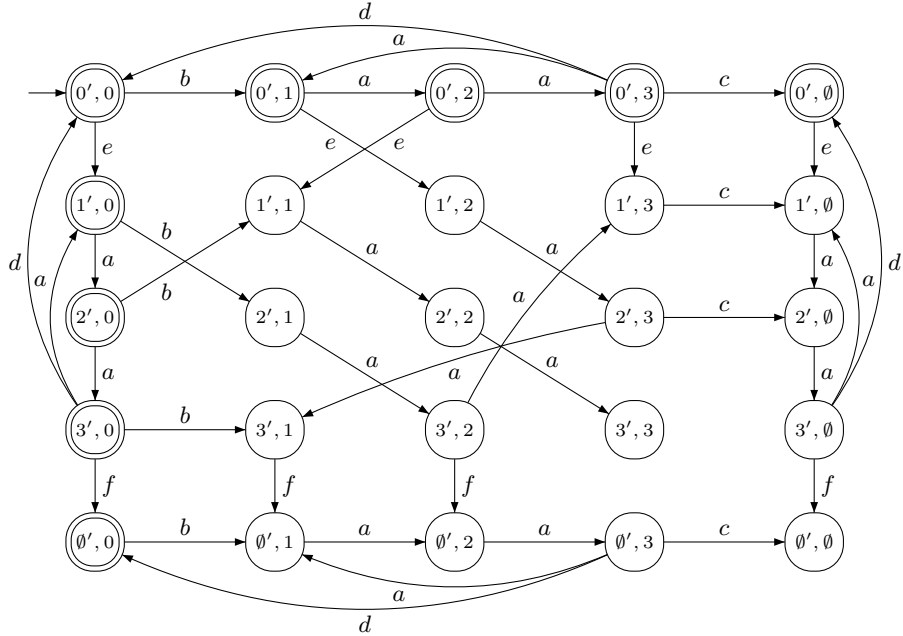
	
By the restricted case, the states of $Q'_m \times Q_n$ are reachable and distinguishable using words in $\{a,b,d,e\}^*$.
Let $R_{\emptyset'} = \{(\emptyset', q) \mid q \in Q_n\}$ and $C_{\emptyset} = \{(p', \emptyset) \mid p' \in Q'_m\}$.
States of $R_{\emptyset'} \cup C_{\emptyset} \cup \{ (\emp',\emp)\}$ are easily seen to be reachable using $c$ and $f$ in addition to $a$, $b$, $d$, and $e$.
We check that the states of $R_{\emptyset'} \cup C_{\emptyset} \cup \{ (\emp',\emp) \}$ are pairwise distinguishable and distinguishable from the states of $Q'_m \times Q_n$.

{\bf Union}
Non-final states of $R_{\emptyset'}$ are distinguished by words in $a^*d$, and those of $C_{\emptyset}$ are similarly distinguishable.
All states besides $\{ (\emp',\emp)\}$ are non-empty since each accepts a word in $\{a,b,d,e\}^*$.
States of $R_{\emptyset'} \cup \{(\emptyset', \emptyset)\}$ are distinguishable from all other states since every other state accepts a word in $ca^*d$;
states of $C_{\emptyset} \cup \{(\emptyset', \emptyset)\}$ are similarly distinguishable from all other states by words in $fa^*d$.
Hence all $(m+1)(n+1)$ states are pairwise distinguishable.

{\bf Symmetric Difference}
Same as union.

{\bf Difference}
States of $R_{\emptyset'} \cup \{ (\emp',\emp)\}$ are empty and therefore equivalent.
However, since the alphabet of $L'_m(a,b,c,d,e) \setminus L_n(a,e,f,d,b)$ is $\{a,b,c,d,e\}$ we can omit $f$ and delete the states of $R_{\emptyset'} \cup \{(\emp',\emp)\}$, and be left with a DFA over $\{a,b,c,d,e\}$ that recognizes $L'_m(a,b,c,d,e) \setminus L_n(a,e,f,d,b)$.
States of $C_{\emptyset}$ are distinguished by words in $a^*d$, and
states of $Q'_m \times Q_n$ are distinguished from states of $C_{\emptyset}$ by words in $a^*da^*d$.
Hence the $mn+m$ remaining states are pairwise distinguishable.

{\bf Intersection}
States of $R_{\emptyset'} \cup C_{\emptyset} \cup \{ (\emp',\emp)\}$ are empty and therefore equivalent.
However, since the alphabet of $L'_m(a,b,c,d,e) \cap L_n(a,e,f,d,b)$ is $\{a,b,d,e\}$, we can omit $c$ and $f$ and delete the states of $R_{\emptyset'} \cup C_{\emptyset} \cup \{(\emp',\emp)\}$, and be left with a DFA over $\{a,b,d,e\}$ that recognizes $L'_m(a,b,c,d,e) \cap L_n(a,e,f,d,b)$. By the restricted case, all $mn$ states are pairwise distinguishable.
\qed
\ee

\ee
\end{proof}

\section{Suffix-Free Languages}
The complexity of suffix-free languages was studied in detail in~\cite{BLY12,BrSz15a,BrSz15,CmJi12,HaSa09,JiOl09}.
For completeness we present a short summary of some of those results.
The main result of~\cite{BrSz15a,BrSz15} is a proof that \emph{a most complex suffix-free language does not exist.} 
Since every suffix-free language has an empty quotient, the restricted and unrestricted cases for binary operations coincide.  

For $n\ge 6$, the transition semigroup of the DFA defined below is the largest transition semigroup of a minimal DFA accepting a suffix-free language. 

\begin{definition}
\label{def:D6}
For $n\ge 4$, 
we define the DFA 
$\cD_n(a,b,c,d,e) =(Q_n,\Sig,\delta,0,F),$
where $Q_n=\{0,\ldots,n-1\}$, $\Sig=\{a,b,c,d,e\}$, 
 $\delta$ is defined by the transformations
$a\colon (0\to n-1) (1,\ldots,n-2)$,
$b\colon (0\to n-1) (1,2)$,
$c\colon (0\to n-1) (n-2\to 1)$,
$d\colon (\{0,1\}\to n-1)$,
$e\colon (Q\setminus \{0\} \to n-1)(0\to 1)$,
and $F=\{q\in Q_n\setminus \{0,n-1\} \mid q \text{ is odd}\}$.
For $n=4$, $a$ and $b$ coincide, and we can use $\Sig=\{b,c,d,e\}$.
Let the transition semigroup of $\cD_n$ be $\Wsf(n)$.
\end{definition}

The main result for this witness is the following theorem:
\begin{theorem}[Semigroup, Quotients, Reversal,  Atoms, Boolean Ops.]
\label{thm:witness6}
Consider DFA  $\cD_n(a,b,c,d,e)$ of Definition~\ref{def:D6}; its language 
$L_n(a,b,c,d,e)$ is a suffix-free language of complexity $n$.
Moreover, it meets the following bounds:
\be
\item
For $n\ge 6$, $L_n(a,b,c,d,e)$ meets the bound $(n-1)^{n-2}+n-2$ for syntactic complexity, and at least five letters are required to reach this bound.
\item
The quotients of $L_n(a,-,-,-,e)$ have complexity $n-1$, except for $L$ which has complexity $n$, and  the empty quotient which has complexity 1. 
\item
For $n\ge 4$, the reverse of $L_n(a,-,c,-,e)$ has complexity $2^{n-2}+1$, and $L_n(a,-,c,-,e)$ has $2^{n-2}+1$ atoms.
\item
Each atom $A_S$ of $L_n(a,b,c,d,e)$ has maximal complexity:
\begin{equation*}
\label{eq:number_of_atoms}
	\kappa(A_S) =
	\begin{cases}
		 2^{n-2}+1,		& \text{if $S=\emp$;} \\

		n, 			& \text{if $S=\{0\}$;}\\
		1 + \sum_{x=1}^{|S|}\sum_{y=0}^{n-2-|S|}\binom{n-2}{x}\binom{n-2-x}{y},
		 			& \emp\neq S\subseteq \{1,\ldots,n-2\}.
		\end{cases}
\end{equation*}

\item
For $n,m \ge 4$, the complexity of $L_m(a,b,-,d,e)\circ L_n(b,a,-,d,e)$ is $mn - (m+n-2)$  if $\circ \in \{\cup, \oplus\}$, 
$mn - (m+2n-4)$ if $\circ = \setminus$, and 
$mn - 2(m+n-3)$ if $\circ = \cap$.
\item
A language which has a subsemigroup of $\Wsf(n)$ as its syntactic semigroup cannot meet the bounds for star and product.
\ee
\end{theorem}

The DFA defined below has the largest transition semigroup when $n\in \{4,5\}$.
The transition semigroup of this DFA is $\Vsf(n)$, and at least $n$ letters are required to generate it.

\begin{definition}
\label{def:D5}
For $n\ge 4$,  
$\cD_n(a,b,c_1,\dots,c_{n-2}) =(Q_n,\Sig_n,\delta,0,\{n-2\}),$
where $Q_n=\{0,\ldots,n-1\}$, $\Sig_n=\{a,b,c_1,\dots,c_{n-2}\}$, 
 $\delta$ is given by 
$a\colon (0\to n-1) (1,\ldots,n-2)$,
$b\colon (0\to n-1) (1,2)$,
and  $c_p\colon (p\to n-1) (0\to p)$ for $1\le p\le n-2$.
\end{definition}

We now define a DFA based on  Definition~\ref{def:D5}, but with only three inputs.

\begin{definition}
\label{def:suffree}
For $n\ge 4$, 
define the DFA 
$\cD_n =(Q_n,\Sig,\delta,0,\{n-2\}),$
where $Q_n=\{0,\ldots,n-1\}$, $\Sig=\{a,b,c\}$, 
and $\delta$ is defined by 
$ a \colon (0 \to n-1) (1,\dots,n-2) $,
$ b \colon (0 \to n-1) (1, 2)$, 
$ c \colon  (1,n-1) (0 \to 1) $.
See Figure~\ref{fig:suffree}.
\end{definition}

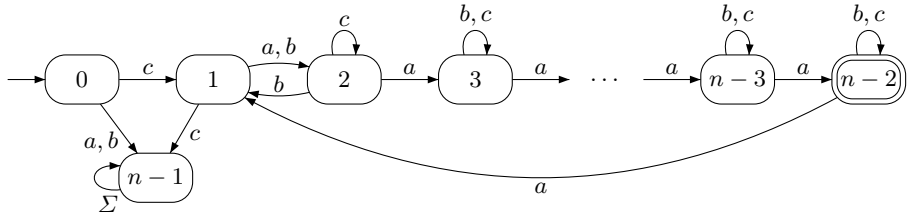
\begin{figure}[ht]
\unitlength 6.2pt
\begin{center}\begin{picture}(40,12)(5,-1)
\gasset{Nh=3,Nw=4.5,Nmr=1.25,ELdist=0.3,loopdiam=1.5}
\node(0)(1,7){0}\imark(0)
\node(1)(9,7){1}
\node(2)(17,7){2}
\node(3)(25,7){3}
\node[Nframe=n](3dots)(33,7){$\dots$}
\node(n-3)(41,7){$n-3$}
\node(n-2)(49,7){$n-2$}\rmark(n-2)
\node(n-1)(5.5,1){$n-1$}

\drawedge(0,1){$c$}
\drawedge[curvedepth=1](1,2){$a,b$}
\drawedge[curvedepth=1,ELdist=-1.25](2,1){$b$}
\drawloop(2){$c$}
\drawloop(3){$b,c$}
\drawloop(n-3){$b,c$}
\drawloop(n-2){$b,c$}
\drawedge(2,3){$a$}
\drawedge(3,3dots){$a$}
\drawedge(3dots,n-3){$a$}
\drawedge(n-3,n-2){$a$}
\drawedge[curvedepth=6,ELdist=.3](n-2,1){$a$}
\drawedge[ELdist= -2.6](0,n-1){$a,b$}
\drawedge(1,n-1){$c$}
\drawloop[loopangle=180
,ELpos=25](n-1){$\Sigma$}
\end{picture}\end{center}
\caption{Witness for star, product, and boolean operations.}
\label{fig:suffree}
\end{figure}

\begin{theorem}[Star, Product, Boolean Operations]
\label{thm:suffree}
Let $\cD_n(a,b,c)$ be the DFA of Definition~\ref{def:suffree}, and let the language it accepts be
$L_n(a,b,c)$. 
Then $L_m$ and its permutational dialects meet the bounds for star, product, and boolean operations as follows:
\be
\item
For $n\ge 4$, $(L_n(a,b,c))^*$ meets the bound $2^{n-2}+1$.
\item
For $m,n\ge 4$, $L'_m(a,b,c) L_n(c,a,b)$ meets the bound $(m-1)2^{n-2}+1$.
\item
For $m,n \ge 4$, but $(m,n)\not= (4,4)$, the complexity of $L'_m(a,b,c)\circ L_n(b,a,c)$ is $mn - (m+n-2)$  if $\circ\in \{\cup,\oplus\}$, 
 $mn - (m+2n-4)$ if $\circ=\setminus$, and 
$mn - 2(m+n-3)$ if $\circ=\cap$.
\ee
\end{theorem}

\begin{proof}
The upper bounds for these operations were established in~\cite{BJZ14}.
\be
\item

{\bf Star}
We will prove that the DFA of Figure~\ref{fig:suffree} meets the bound $2^{n-2}+1$.
Since there are no incoming transitions to the initial state $0$, to obtain an NFA accepting $L_n^*$ it is sufficient to make state $0$ final, add an $\eps$-transition from state $n-2$ to state $0$, and delete state $n-1$.
We will show that in the subset construction the following states are reachable and pairwise distinguishable: $\{0\}$, any one of the $2^{n-3}$ subsets $S$ of $P=\{1,\dots,n-3\}$, and $2^{n-3}$ subsets of the form $\{0,n-2\}\cup S$, where $S\subseteq P$.

The initial state is $\{0\}$, the empty set is reached by $a$,  set $\{q\}$, $q\in P$ by $ca^{q-1}$, and $\{0, n-2\}$ by $ca^{n-3}$.
Notice that the set $\{1, 2, \dots, k\}$, where $2 \le k \le n-3$, is reached from $\{1, \dots, k-1\}$ by $a^{n-1-k}ca^{k-1}$.
It is well known that $a$ and $b$ generate all permutations of $\{1, \dots, n-2\}$; hence, for any $S \subseteq P$ there is a word $w \in \{a,b\}^*$ such that $\{1, 2, \dots, |S|\}w = S$.
Similarly a set $\{0, q_1, q_2, \dots, q_k, n-2\}$ with $k \le n-4$ is reached by a permutation from $\{1, 2, \dots, k+1\}$.
Finally, $\{0, 1, \dots, n-2\}$ is reached from $\{1, \dots, n-3\}$ by $ac$; hence every subset of $P$ is reachable, as are the sets $S \cup \{0, n-2\}$ for $S \subseteq P$.

Any pair of sets which differ by $q \in \{1, \dots, n-2\}$ is distinguished by $a^{n-2-q}$, and the empty set is distinguished from $\{0\}$ by $ca^{n-3}$. Thus, all $2^{n-2}+1$ sets are reachable and pairwise distinguishable.

\item
{\bf Product}
We construct an NFA for the product by deleting state $(m-1)'$, adding an $\eps$-transition from state $(m-2)'$ to state 0, deleting state $n-1$, and making state $n-2$ the only final state.
We will show that the following sets are all reachable and pairwise distinguishable:
$\{0'\}$, $\{p'\} \cup S$, $ 1 \le p < m-2$, $\{ (m-2)', 0 \} \cup S$, and $S$, where 
$S\subseteq Q_{n-1}\setminus \{0\}$.

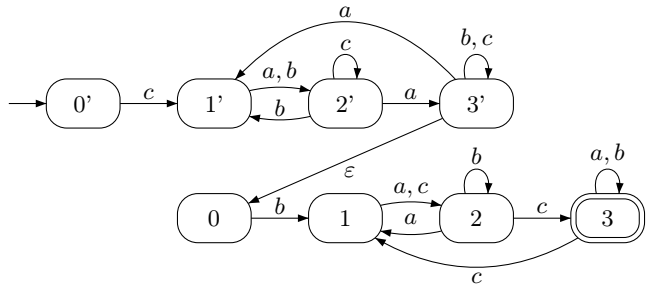
\begin{figure}[ht]
\unitlength 6.2pt
\begin{center}\begin{picture}(40,15)(0,3)
\gasset{Nh=3,Nw=4.5,Nmr=1.25,ELdist=0.3,loopdiam=1.5}
\node(0')(1,12){0'}\imark(0')
\node(1')(9,12){1'}
\node(2')(17,12){2'}
\node(3')(25,12){3'}

\node(0)(9,5){0}
\node(1)(17,5){1}
\node(2)(25,5){2}
\node(3)(33,5){3}\rmark(3)

\drawedge(0',1'){$c$}
\drawedge[curvedepth=1](1',2'){$a,b$}
\drawedge[curvedepth=1,ELdist=-1.25](2',1'){$b$}
\drawloop(2'){$c$}
\drawloop(3'){$b,c$}
\drawedge(2',3'){$a$}
\drawedge[curvedepth=-5,ELside=r](3',1'){$a$}

\drawedge(0,1){$b$}
\drawedge[curvedepth=1](1,2){$a,c$}
\drawedge[curvedepth=1,ELdist=-1.25](2,1){$a$}
\drawloop(2){$b$}
\drawloop(3){$a,b$}
\drawedge(2,3){$c$}
\drawedge[curvedepth=3](3,1){$c$}

\drawedge(3',0){$\eps$}

\end{picture}\end{center}
\caption{NFA for product $L'_{5}(a,b,c)L_{5}(c,a,b)$ of suffix-free languages.}
\label{fig:suffreeprod}
\end{figure}

State $\{0'\}$ is initial, $\{p'\}$ is reached by $ca^{p-1}$ for $1 \le p < m-2$, and $\{(m-2)',0\}$ is reached by $ca^{m-3}$. 
From $\{(m-2)', 0, q_2 - q_1, \dots, q_k - q_1\}$ we reach $\{(m-2)', 0, q_1, q_2, \dots, q_k\}$ by $cbc^{q_1 - 1}$;
hence $\{(m-2)', 0\} \cup S$ is reachable for any $S \subseteq Q_{n-1} \setminus \{0\}$.
Now for any $p' \in \{1, \dots, (m-3)'\}$, if $p$ is even, then $\{p'\} \cup S$ is reached from $\{(m-2)',0\} \cup S$ by $a^p$, and
if $p$ is odd, then $\{p'\} \cup S$ is reached from $\{(m-2)',0\} \cup (Sa)$ by $a^p$.
Finally, $S$ is reached from $\{1'\} \cup S$ by $c^{n-2}$.

Any two states which differ on some $q \in Q_{n-1} \setminus \{0\}$ are distinguished by $c^{n-2-q}$.
Two states that differ on $p' \in Q'_m \setminus \{0'\}$ are distinguished by $a^{m-2-p}bc^{n-3}$.
Finally, $\{0'\}$ is distinguishable from all other states because it is the only state that accepts $ca^{m-3}bc^{n-3}$.

\item
{\bf Boolean Operations}
We consider the direct product of $\cD'_m(a,b,c)$ and $\cD_n(b,a,c)$, where $m,n \ge 4$ and $(m,n) \not= (4,4)$, illustrated in Figure~\ref{fig:sfreecross}.
Let $S = \{(p',q) \mid p' \in Q'_{m-1} \setminus \{0'\}, q \in Q_{n-1} \setminus \{0\}\}$,
$R = \{((m-1)', q) \mid q \in Q_{n-1} \setminus \{0\}\}$,
and $C = \{(p', n-1) \mid p \in Q'_{m-1} \setminus \{0'\}\}$.

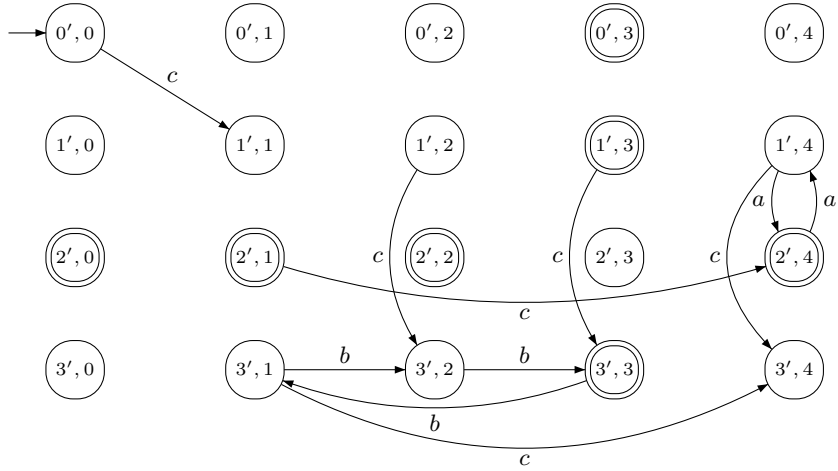
\begin{figure}[ht]
\unitlength 8.5pt
\begin{center}\begin{picture}(35,19)(0,-2.5)
\gasset{Nh=2.6,Nw=2.6,Nmr=1.2,ELdist=0.3,loopdiam=1.2}
	{\scriptsize
\node(0'0)(2,15){$0',0$}\imark(0'0)
\node(1'0)(2,10){$1',0$}
\node(2'0)(2,5){$2',0$}\rmark(2'0)
\node(3'0)(2,0){$3',0$}

\node(0'1)(10,15){$0',1$}
\node(1'1)(10,10){$1',1$}
\node(2'1)(10,5){$2',1$}\rmark(2'1)
\node(3'1)(10,0){$3',1$}

\node(0'2)(18,15){$0',2$}
\node(1'2)(18,10){$1',2$}
\node(2'2)(18,5){$2',2$}\rmark(2'2)
\node(3'2)(18,0){$3',2$}

\node(0'3)(26,15){$0',3$}\rmark(0'3)
\node(1'3)(26,10){$1',3$}\rmark(1'3)
\node(2'3)(26,5){$2',3$}
\node(3'3)(26,0){$3',3$}\rmark(3'3)

\node(0'4)(34,15){$0',4$}
\node(1'4)(34,10){$1',4$}
\node(2'4)(34,5){$2',4$}\rmark(2'4)
\node(3'4)(34,0){$3',4$}
	}

\drawedge(0'0,1'1){$c$}

\drawedge[curvedepth=-2, ELside=r](1'2,3'2){$c$}
\drawedge[curvedepth=-2, ELside=r](1'3,3'3){$c$}
\drawedge[curvedepth=-3, ELside=r](1'4,3'4){$c$}
\drawedge[curvedepth=-2, ELside=r](2'1,2'4){$c$}
\drawedge[curvedepth=-3.5, ELside=r](3'1,3'4){$c$}

\drawedge(3'1,3'2){$b$}
\drawedge(3'2,3'3){$b$}
\drawedge[curvedepth=1.7](3'3,3'1){$b$}
\drawedge[curvedepth=-1, ELside=r](1'4,2'4){$a$}
\drawedge[curvedepth=-1, ELside=r](2'4,1'4){$a$}
\end{picture}\end{center}
\caption{Direct product for $L'_4(a,b,c)\oplus L_5(b,a,c)$ shown partially.}
\label{fig:sfreecross}
\end{figure}

We first determine which states are reachable in the direct product.
State $(0',0)$ is initial and $(1',1)$ is reachable by $c$.
By \cite[Theorem 1]{BBMR14} and computation for the cases $(m,n) \in \{(5,6),(6,5),(6,6)\}$, all $(m-2)(n-2)$ states of $S$ are reachable from $(1',1)$ for all $m,n \ge 4$, $(m,n) \not= (4,4)$.
State $((m-1)', n-2)$ is reached from $(1', n-2)$ by $c$, and $((m-1)', q)$, $q \in Q_{n-1} \setminus \{0\}$, is reached from $((m-1)', n-2)$ by $b^{q}$; thus states of $R$ are reachable.
Similarly, state $((m-2)', n-1)$ is reached from $((m-2)', 1)$ by $c$, and $(p', n-1)$, $p \in Q'_{m-1} \setminus \{0'\}$, is reached from $((m-2)', n-1)$ by $a^{p}$; thus states of $C$ are reachable.
State $((m-1)',n-1)$ is reachable from $((m-1)',1)$ by $c$.
States $(p', 0)$ for $p' \not= 0'$ and $(0', q)$ for $q \not= 0$ are not reachable in the direct product, leaving $mn-(m+n-2)$ reachable states.

We check distinguishability for each operation.
In all cases, $(0', 0)$ is distinguishable from every other state because it is non-empty, and it goes to the empty state by $a$.
Again using \cite[Theorem 1]{BBMR14} and computation for the cases $(m,n) \in \{(5,6),(6,5),(6,6)\}$, the states of $S$ are pairwise distinguishable for all four operations.

{\bf Union} States of $R$ are distinguished from each other by words in $b^*$, and from states of $S$ by words in $\{a,b\}^*$.
Similarly, states of $C$ are distinguished from each other by words in $a^*$, and from states of $S$ by words in $\{a,b\}^*$.
States of $R$ are distinguished from those of $C$ by words in $a^*$.
Hence the $mn-(m+n-2)$ reachable states are pairwise distinguishable.

{\bf Symmetric Difference} Same as union.

{\bf Difference} The states of $C$ are all empty, and hence equivalent to $((m-1)', n-1)$.
States of $R$ are pairwise distinguishable by words in $a^*$, and they are distinguished from states of $S$ by words in $\{a,b\}^*$.
Hence there are $mn-(m+2n-4)$ distinguishable states.

{\bf Intersection}
The states of $R \cup C$ are all empty, and hence there are only $mn -2(m+n-3)$ distinguishable states.
\qed
\ee
\end{proof}

The transition semigroup of the DFA of Definition~\ref{def:D5s} is a also a subsemigroup of $\Vsf(n)$, and its language also meets the bounds for product, star and boolean operations. The advantage of this DFA is that its witnesses use only two letters for star
and only two letters (but three transformations) for boolean operations. Its disadvantages are the rather complex transformations.
For more details see~\cite{BrSz15}. The DFA of Definition~\ref{def:suffree} seems to us more natural.

\begin{definition}
\label{def:D5s}
For $n\ge 6$, 
we define the DFA 
$\cD_n =(Q_n,\Sig,\delta,0,\{1\}),$
where $Q_n=\{0,\ldots,n-1\}$, $\Sig=\{a,b,c\}$, 
and $\delta$ is defined by the transformations
$ a \colon (0 \to n-1) (1,2,3) (4,\dots,n-2) $,
$ b \colon (2 \to n-1) (1 \to 2) (0 \to 1) (3,4) $, 
$ c \colon (0 \to n-1) (1,\dots, n-2)  $.
\end{definition}

\section{Conclusions}
We have examined the complexity properties of left-ideal, suffix-closed, and suffix-free languages together because they are all special cases of suffix-convex languages.
We have used the same most complex regular language as a basic component in all three cases.

 Our results are summarized in Table~\ref{tab:summary}. The largest bounds are shown in boldface type.
Recall that for regular languages we have the following results: 
semigroup: $n^n$; 
reverse: $2^n$; 
star: $2^{n-1}+2^{n-2}$; 
restricted product:  $(m-1)2^n+2^{n-1}$;
unrestricted product: $m2^n+2^{n-1}$;  
restricted $\cup$ and $\oplus$: $mn$; 
unrestricted $\cup$ and $\oplus$: $(m+1)(n+1)$;
restricted $\setminus$: $mn$;
unrestricted $\setminus$: $mn+m$;
restricted $\cap$: $mn$; 
unrestricted $\cap$: $mn$. 
 \renewcommand{\arraystretch}{1.3}
\begin{table}[ht]
\caption{Complexities of special suffix-convex languages}
\label{tab:summary}

\footnotesize
\begin{center}
$
 \begin{array}{|c||c|c|c||}    
\hline
 & \ \text{Left-Ideal} \ & \ \text{Suffix-Closed}  \ &  \ \text{Suffix-Free}\     \\
\hline \hline
\ Semigroup \   	
 &\ \mathbf{n^{n-1}+n-1} \	&\mathbf{n^{n-1}+n-1} &(n-1)^{n-2} +n-2  \\
\hline
\ Reverse \   
 &\ \mathbf{2^{n-1}+1} \	& \mathbf{2^{n-1}+1} &\ 2^{n-2} +1 \ \\
\hline
\ Star \   
&\ n+1 \	& \ n \ &\ \mathbf{2^{n-2}+1} \   \\
\hline
\ Product\;  restricted\   &
\ m+n-1 \	& \ {mn-n+1} \ &\  \mathbf{(m-1)2^{n-2}+1} \  \\
\hline
\ Product\; unrestricted\   &
\ mn+m+n \	& \ mn+m+1 \ &\ \mathbf{(m-1)2^{n-2}+1} \    \\

\hline
\ \cup  \; restricted \   &	
\ \mathbf{mn} \	& \ \mathbf{mn} \ &\ mn-(m+n-2) \    \\
\hline
\ \cup \;   unrestricted \   &	
\ \mathbf{(m+1)(n+1)} \	& \ \mathbf{(m+1)(n+1)} \ &\ mn-(m+n-2) \    \\
 \hline
\ \oplus\;  restricted \   &	
\mathbf{mn} 	& \ \mathbf{mn} \ &\ mn-(m+n-2) \    \\
 \hline
\ \oplus\;   unrestricted \   &	
 \mathbf{(m+1)(n+1)} 	& \ \mathbf{(m+1)(n+1)} \ &\ mn-(m+n-2) \    \\
 \hline
\setminus\; restricted \   &	
\mathbf{mn} \	& \ \mathbf{mn} \ & mn-(m+2n-4)   \\
\hline
\setminus\; unrestricted \   &	
 \mathbf{mn+m} \	& \ \mathbf{mn+m} \ & mn-(m+2n-4)    \\
\hline
\cap\; restr. \text{ and } unrestr. \   &	
\ \mathbf{mn} \	&\ \mathbf{mn} \ & mn-2(m+n-3)    \\
\hline
\end{array} 
$
\end{center}
\end{table}

\providecommand{\noopsort}[1]{}


\begin{thebibliography}{10}
\providecommand{\url}[1]{\texttt{#1}}
\providecommand{\urlprefix}{URL }

\bibitem{AnBr09}
Ang, T., Brzozowski, J.: Languages convex with respect to binary relations, and
  their closure properties. Acta Cybernet.  19(2),  445--464 (2009)

\bibitem{BBMR14}
Bell, J., Brzozowski, J., Moreira, N., Reis, R.: Symmetric groups and quotient
  complexity of boolean operations. In: Esparza, J., et~al. (eds.) ICALP 2014.
  LNCS, vol. 8573, pp. 1--12. Springer (2014)

\bibitem{BPR10}
Berstel, J., Perrin, D., Reutenauer, C.: Codes and Automata (Encyclopedia of
  Mathematics and its Applications). Cambridge University Press (2010)

\bibitem{Brz10}
Brzozowski, J.: Quotient complexity of regular languages. J. Autom. Lang. Comb.
   15(1/2),  71--89 (2010)

\bibitem{Brz13}
Brzozowski, J.: In search of the most complex regular languages. Int. J. Found.
  Comput. Sci.,  24(6),  691--708 (2013)

\bibitem{Brz16}
Brzozowski, J.: Unrestricted state complexity of binary operations on regular
  languages. In: C\^ampeanu, C., et. al (eds.) DCFS 2016. LNCS, vol. 9777, pp.
  60--72. Springer (2016)

\bibitem{BrDa15}
Brzozowski, J., Davies, S.: Quotient complexities of atoms in regular ideal
  languages. Acta Cybernet.  22(2),  293--311 (2015)

\bibitem{BDL15}
Brzozowski, J., Davies, S., Liu, B.Y.V.: Most complex regular ideal languages
  (October 2015), {\small\tt http://arxiv.org/abs/1511.00157}

\bibitem{BJL13}
Brzozowski, J., Jir{\'a}skov{\'a}, G., Li, B.: Quotient complexity of ideal
  languages. Theoret. Comput. Sci.  470,  36--52 (2013)

\bibitem{BJZ14}
Brzozowski, J., Jir{\'a}skov{\'a}, G., Zou, C.: Quotient complexity of closed
  languages. Theory Comput. Syst.  54,  277--292 (2014)

\bibitem{BLY12}
Brzozowski, J., Li, B., Ye, Y.: Syntactic complexity of prefix-, suffix-,
  bifix-, and factor-free regular languages. Theoret. Comput. Sci.  449,
  37--53 (2012)

\bibitem{BrSi16a}
Brzozowski, J., Sinnamon, C.: Complexity of prefix-convex regular languages
  (2016), {\tt http://arxiv.org/abs/1605.06697}

\bibitem{BrSi16}
Brzozowski, J., Sinnamon, C.: Unrestricted state complexity of binary
  operations on regular and ideal languages (2016), {\tt
  http://arxiv.org/abs/1609.04439}

\bibitem{BrSz14}
Brzozowski, J., Szyku{\l}a, M.: Upper bounds on syntactic complexity of left
  and two-sided ideals. In: Shur, A.M., Volkov, M.V. (eds.) DLT 2014. LNCS,
  vol. 8633, pp. 13--24. Springer (2014)

\bibitem{BrSz15a}
Brzozowski, J., Szyku{\l}a, M.: Complexity of suffix-free regular languages.
  In: Kosowski, A., Walukiewicz, I. (eds.) FCT 2015. LNCS, vol. 9210, pp.
  146--159. Springer (2015)

\bibitem{BrSz15}
Brzozowski, J., Szyku{\l}a, M.: Complexity of suffix-free regular languages
  (2015), {\small\tt http://arxiv.org/abs/1504.05159}

\bibitem{BSY15}
Brzozowski, J., Szyku{\l}a, M., Ye, Y.: Syntactic complexity of regular ideals
  (September 2015), {\small\tt http://arxiv.org/abs/1509.06032}

\bibitem{BrTa13}
Brzozowski, J., Tamm, H.: Quotient complexities of atoms of regular languages.
  Int. J. Found. Comput. Sci.  24(7),  1009--1027 (2013)

\bibitem{BrTa14}
Brzozowski, J., Tamm, H.: Theory of \'atomata. Theoret. Comput. Sci.  539,
  13--27 (2014)

\bibitem{BrYe11}
Brzozowski, J., Ye, Y.: Syntactic complexity of ideal and closed languages. In:
  Mauri, G., Leporati, A. (eds.) DLT 2011. LNCS, vol. 6795, pp. 117--128.
  Springer Berlin / Heidelberg (2011)

\bibitem{CmJi12}
Cmorik, R., Jir{\'a}skov{\'a}, G.: Basic operations on binary suffix-free
  languages. In: Kot{\'a}sek, Z., et~al. (eds.) MEMILCS. pp. 94--102 (2012)

\bibitem{HaSa09}
Han, Y.S., Salomaa, K.: State complexity of basic operations on suffix-free
  regular languages. Theoret. Comput. Sci.  410(27-29),  2537--2548 (2009)

\bibitem{HoKu11}
Holzer, M., Kutrib, M.: Descriptional and computational complexity of finite
  automata—a survey. Information and Computation  209(3),  456 -- 470 (2011)

\bibitem{HoKo04}
Holzer, M., K\"{o}nig, B.: On deterministic finite automata and syntactic
  monoid size. Theoret. Comput. Sci.  327(3),  319--347 (2004)

\bibitem{Iva16}
Iv\'an, S.: Complexity of atoms, combinatorially. Inform. Process. Lett.
  116(5),  356--360 (2016)

\bibitem{JiOl09}
Jir\'askov\'a, G., Olej\'ar, P.: State complexity of union and intersection of
  binary suffix-free languages. In: Bordihn, H., et~al. (eds.) NMCA. pp.
  151--166. Austrian Computer Society (2009)

\bibitem{KLS05}
Krawetz, B., Lawrence, J., Shallit, J.: State complexity and the monoid of
  transformations of a finite set. In: Domaratzki, M., Okhotin, A., Salomaa,
  K., Yu, S. (eds.) CIAA 2005. LNCS, vol. 3317, pp. 213--224. Springer Berlin /
  Heidelberg (2005)

\bibitem{Myh57}
Myhill, J.: Finite automata and representation of events. Wright Air
  Development Center Technical Report  57--624 (1957)

\bibitem{Pin97}
Pin, J.E.: Syntactic semigroups. In: Handbook of Formal Languages, vol.~1:
  Word, Language, Grammar, pp. 679--746. Springer, New York, NY, USA (1997)

\bibitem{Thi73}
Thierrin, G.: Convex languages. In: Nivat, M. (ed.) Automata, Languages and
  Programming, pp. 481--492. North-Holland (1973)

\bibitem{Yu01}
Yu, S.: State complexity of regular languages. J. Autom. Lang. Comb.  6,
  221--234 (2001)

\end{thebibliography}
\end{document}